\documentclass{article}

\usepackage{microtype}
\usepackage{graphicx}
\usepackage{subcaption}
\usepackage{booktabs} 

\usepackage{hyperref}

\usepackage{comment}

 \usepackage[preprint]{icml2026}

\usepackage{amsmath}
\usepackage{amssymb}
\usepackage{mathtools}
\usepackage{amsthm}
\usepackage{bm}

\usepackage[capitalize,noabbrev]{cleveref}

\theoremstyle{plain}
\newtheorem{theorem}{Theorem}[section]

\theoremstyle{definition}

\newtheorem{assumption}[theorem]{Assumption}
\theoremstyle{remark}

\usepackage[textsize=tiny]{todonotes}

\icmltitlerunning{Adaptive Quality-Diversity Trade-offs for Large-Scale Batch Recommendation}

\newcommand{\bbR}{\mathbb{R}}
\newcommand{\bbN}{\mathbb{N}}
\newcommand{\ITEM}[1]{\bm{\phi}^{#1}}
\newcommand{\USER}[1]{\bm{h}^{#1}}

\newcommand{\PREWARD}{\Theta}
\newcommand{\FREWARD}{q_{\PREWARD}}
\newcommand{\MREWARD}[1]{Q_{#1}}

\newcommand{\VOLUME}[1]{\text{vol}(#1)}

\newcommand{\KERNEL}{k}
\newcommand{\NBATCH}{B}
\newcommand{\NITEMS}{N}
\newcommand{\HORIZON}{T}
\newcommand{\HISTORY}[2]{\mathcal{H}^{#2}}
\newcommand{\HISTORYGT}[2]{\bm{\mathcal{H}}^{#2}}
\newcommand{\BATCH}[1]{\mathcal{S}^{#1}}
\newcommand{\DI}{d}
\newcommand{\DU}{m}
\newcommand{\TRACE}{\text{tr}}
\newcommand{\USERKERNEL}[4]{L^{#1}_{#2}(#3; #4)}
\newcommand{\GOODSET}[1]{\mathbb{S}_{t}}
\newcommand{\SETSCORE}[1]{\text{scr}}
\newcommand{\FEEDBACK}[1]{\bm{y}^{#1}}

\newcommand{\EG}{\textit{e}.\textit{g}.,\:}
\newcommand{\IE}{\textit{i}.\textit{e}.,\:}
\newcommand{\ALGO}{B-DivRec}

\begin{document}

\twocolumn[
  \icmltitle{Adaptive Quality-Diversity Trade-offs for Large-Scale Batch Recommendation}



  \icmlsetsymbol{equal}{*}

  \begin{icmlauthorlist}
    \icmlauthor{Cl\'{e}mence R\'{e}da}{aaa}
    \icmlauthor{Tomas Rigaux}{bbb}
    \icmlauthor{Hiba Bederina}{ccc}
    \icmlauthor{Koh Takeuchi}{bbb}
    \icmlauthor{Hisashi Kashima}{bbb}
    \icmlauthor{Jill-J\^{e}nn Vie}{ccc}
  \end{icmlauthorlist}

  \icmlaffiliation{aaa}{Ecole Normale Sup\'{e}rieure PSL, Paris, F-75005, France}
  \icmlaffiliation{bbb}{Kyoto University, Kyoto, J-606-8501, Japan}
  \icmlaffiliation{ccc}{Soda, Inria Paris Saclay, Palaiseau, F-91120, France}

  \icmlcorrespondingauthor{Cl\'{e}mence R\'{e}da}{reda@bio.ens.psl.eu}

  \icmlkeywords{recommender systems, serendipity, diversity, determinantal point processes, online learning}

  \vskip 0.3in
]



\printAffiliationsAndNotice{}  

\begin{abstract}
A core research question in recommender systems is to propose batches of highly relevant and diverse items, that is, items personalized to the user's preferences, but which also might get the user out of their comfort zone. This diversity might induce properties of serendipidity and novelty which might increase user engagement or revenue. However, many real-life problems arise in that case: e.g., avoiding to recommend distinct but too similar items to reduce the churn risk, and computational cost for large item libraries, up to millions of items. First, we consider the case when the user feedback model is perfectly observed and known in advance, and introduce an efficient algorithm called B-DivRec combining determinantal point processes and a fuzzy denuding procedure to adjust the degree of item diversity. This helps enforcing a quality-diversity trade-off throughout the user history. Second, we propose an approach to adaptively tailor the quality-diversity trade-off to the user, so that diversity in recommendations can be enhanced if it leads to positive feedback, and vice-versa. Finally, we illustrate the performance and versatility of B-DivRec in the two settings on synthetic and real-life data sets on movie recommendation and drug repurposing. 
\end{abstract}

\section{Introduction}\label{sec:introduction}

Preserving user engagement, that is, the willingness of users to query a recommender system and to interact with recommended items, is crucial and yet a difficult task. It is well-known that, beyond recommending merely the most popular items or those closest to the estimated user's interests, introducing diversity in recommendation is key to avoid the churn risk, \IE customer attrition~\citep{poulain2020investigating}. This topic has been widely studied under the name of ``diverse/novel recommendations'' or ``serendipity''~\citep{abbassi2009getting,kotkov2018investigating,ziarani2021serendipity}. The rationale behind it appears in several real-life contexts: for instance, diversity might increase revenue by keeping user engagement high in the movie streaming or music industry~\citep{van2013deep,anderson2020algorithmic}; get a teenager out of their comfort zone and make them discover new cultural goods~\citep{ibrahim2025diversified}; evaluate the global state of a student's knowledge on a specific subject in education~\citep{chavarriaga2014recommender,yanes2020machine}; or discover a first-in-class drug treatment in the pharmaceutical industry, where it has been shown that first-in-class molecules might generate higher revenue compared to well-known classes of molecules with a therapeutic advantage, called ``best-in-class''~\citep{schulze2013matters,spring2023first}. All in all, the goal in recommender systems is to satisfy apparently contradictory objectives: to recommend user-personalized items which introduce diversity in the user's history of recommended items.

In practice, 
one must also face the problems of developing computationally tractable pipelines on large item libraries
~\citep{cha2018drug}. Finally, 
the definition of diversity and quantifying the diversity in the recommended (batch of) items is tricky. 
Diversity might be understood in terms of ``intrabatch'' (local) or ``interbatch'' (global)~\citep{bederina2025bayesian}, respectively meaning that each batch of recommended items should be diverse, or the batch of recommended items should be diverse with respect to the user's prior history of selected items. \citet{xu2024serendipitous} 
tackles a novel problem of the literature, that is, 
tuning the level of diversity 
to the user based on a prior training set. However, this is not available at cold start. Tolerance to diversity might vary across users, and contributes to the user engagement. 
As we discuss in the next section, our work is one of the few works handling global diversity
~\citep{affandi2012markov,bederina2025bayesian,ibrahim2025diversified}. Furthermore, no other paper before has proposed an automated \textit{online} procedure for 
tuning the level of diversity. 


\section{Related works}\label{sec:related_works}

The topics addressed in this paper are related to a large section of the literature on recommendation (serendipity/out-of-the-box recommendation, avoidance of the churn risk, tractable recommender systems). We defer the reader to Appendix~\ref{app:related_works} for a more comprehensive review of the literature. 

\paragraph{Quality-Diversity Tradeoff metrics.} A closer look at the literature shows that there is no consensus metric capturing the quality-diversity tradeoff, sometimes called novelty. The proposed metrics are either use case-dependent (\EG user enjoyment in~\citet{wang2025beyond} for video recommendation), or only applicable with a predefined item similarity metric or user feedback model. This considerably restricts their reusability~\citep{abbassi2009getting,ziarani2021serendipity,kotkov2024dark,xu2024serendipitous}. A generic and intuitive tradeoff metric--instead of evaluating quality and diversity separately--is still missing. We turn to determinantal point processes to guide this tradeoff.

\paragraph{Determinantal Point Processes.} A point process is a distribution over finite subsets of a (finite) set $\Omega$. A determinantal point process (DPP) is a process where the probability of sampling a subset $A$ is correlated to the determinant of the kernel function $\KERNEL$ applied to this subset, that is, $\mathbb{P}(\BATCH{} = A) \propto \det \KERNEL_{A,A}$~\citep{macchi1975coincidence}. Here, we consider the definition of DPPs with $L$-ensembles, as described in~\citet{borodin2005eynard}, and $\KERNEL_{\Omega,\Omega}$ is a positive semi-definite matrix. Intuitively, the higher the volume formed by the item embeddings in this subset, the higher the probability. 
Conditioning over another subset $\HISTORY{}{}$ of items--that is, sampling a subset which is diverse compared to previously selected set of points--can be described with the following distribution~\citep{borodin2005eynard} $\mathbb{P}(\BATCH{} = A \cup \HISTORY{}{} \mid \BATCH{} \supseteq \HISTORY{}{}) \propto \det\big(\KERNEL_{A,A}-\KERNEL_{A,\HISTORY{}{}} \KERNEL_{\HISTORY{}{},\HISTORY{}{}}^{-1} \KERNEL_{A,\HISTORY{}{}}^\intercal \big)$. This is a simple approach to integrating the user history to the recommendation
. However, conditioning has a dependency in $\Omega(|\HISTORY{}{}|^3)$ where $\HISTORY{}{}$ is the user history, and inversion of the history-related kernel matrix might be expensive. 
To tackle the issue of recommending items with high quality/relevance and high ``intrabatch''/local diversity, \citet{kulesza2010structured} introduced the quality-diversity (QD) decomposition of a DPP. Given $\NITEMS$ item embeddings $\{\ITEM{i}\}_{i \in \Omega}$ such that $\|\ITEM{i}\|_2=1$ for any $i \in \Omega$, and positive quality scores 
for each item and a fixed user, 
the QD decomposition is the distribution $\mathbb{P}(\BATCH{} = A ) \propto \det(Q_{A} \Phi_{A} \Phi_{A}^\intercal Q_{A})$, where $Q_{A}$ is a diagonal matrix of size $\NBATCH \times \NBATCH$ which contains the quality scores for each item $i_1,\dots,i_\NBATCH$ in $A$, and $\Phi_{A} := [\ITEM{i_1},\dots,\ITEM{i_\NBATCH}]^\intercal$ is the row-concatenation of all item embeddings in $A$. Many papers relied on this approach~\citep{gong2014diverse,wilhelm2018practical,zhan2021multiple,svensson2025diverse,xuan2025diverse}. However, this decomposition might be 
restrictive, as the control of the quality-diversity trade-off is limited, and only linear or RBF~\citep{affandi2014learning,wilhelm2018practical} kernels are considered.  \citet{affandi2012markov} is closer to our objective of increasing the ``interbatch'' or global diversity across consecutive recommendation rounds. Authors introduce Markov DPPs, which were applied to the daily recommendation of news headlines. The main idea is to condition the subset of items sampled at round $t+1$ on the subset sampled at round $t$. 
However, their procedure is 
not tractable even for moderately large data sets, 
as shown 
in Section~\ref{sec:experiments}. 

\paragraph{Greedy recommender systems for (local) diversity.} Besides DPPs, there is a plethora of other approaches, many of them~\citep{li2024contextual} relying or comparing to Maximal Marginal Relevance (MMR)~\citep{carbonell1998use}. The problem originally tackled by MMR is slightly different from ours: given a library of items, an item similarity function and a quality-diversity trade-off parameter $\lambda \in [0,1]$, the goal is to return an item $i_t$ at round $t$ which is both relevant to the query item $\USER{\text{qr}}$ and diverse with respect to previously selected items $i_1,i_2,\dots,i_{t-1}$. 
However, it is easy to rewrite the expression in MMR to fit the setting where a user context and history is provided as query. Instead of a similarity function, we consider a kernel function $\KERNEL$. For $b \in \{1,\dots,\NBATCH\}$, at time $t$, MMR selects item 
\[i^{t+b}  :=  \arg\max_{i \in \Omega} \lambda \ R_{i,t} 
-   \max_{j \in \HISTORY{\USER{t}}{t}\cup\{i^{t},.,i^{t+b-1}\}} (1-\lambda) \KERNEL(\ITEM{i},\ITEM{j}) \;.\] 
Yet, MMR only performs item-pairwise comparisons and greedily builds the set of recommendations, while 
the literature on DPPs could be leveraged to sample a whole diverse subset with a complexity linear in the number of items. 

\paragraph{User intent-aware models.} A more recent subfield considers neural slate optimization or reinforcement learning models, which implement a diversification/whole-page optimization step after one or two relevance-based filtering steps on items. The goal is to retrieve relevant items that maximally cover different aspects of the user's tastes or query. However, \citet{carraro2025enhancing} relies on LLMs, which induces a privacy issue and lacks control of the tradeoff. Furthermore, the methods in~\citet{liu2023personalized,roy2025don,wang2025beyond} are meant to be applied at small scale on a prefiltered set of candidate items, as they have time complexities polynomial in $N$, and require the definition of behavioral sequences~\citep{liu2023personalized}, subqueries~\citep{santos2010exploiting} or categorical variables~\citep{zhang2021conditional} to cover, which restricts their flexibility. Conversely, our framework only requires item embeddings/feature vectors, which should contain all information needed to assess diversity.

\section{Setting: Quality-Diversity Tradeoff (QDT)}\label{sec:setting}

\paragraph{Problem setting.} We denote $\NITEMS$ the total number of items (potentially of the order of millions) in the universe $\Omega := \{1,2,\dots,\NITEMS\}$. Recommendation rounds happen sequentially: (1) a (possibly new) user $\USER{t}$ queries the recommender system at each time $t$, (2) the recommender system returns a fixed-size batch $\BATCH{t} \subset \Omega$ of $\NBATCH$ items to present to the user, (3) user $\USER{t}$ outputs feedback values $\FEEDBACK{t} := \{y^t_1,\dots,y^t_\NBATCH\}$ for each recommendation. We choose to ignore the specific identity of items and users, and instead, respectively define them with the item embedding $\ITEM{i} \in \bbR^\DI$ for item $i$, or by the \textit{summary} of user history $\USER{t} \in \bbR^\DU$ at each recommendation round $t$. Ideally, $\USER{t}$ captures all the information about the items the user has \textit{previously interacted with} (\EG $\USER{t}$ can describe how much one of the $\DU$ categories is liked). It can be hard-coded (one-hot encoding of liked item categories), or an embedding learned over previously collected items, or the result of some dimension reduction algorithm applied to the embeddings of positively interacted with items. We denote $\HISTORY{\USER{}}{t}$ the user history of prior recommendations to user $\USER{t}$ up to the round $t$ (not included). We define item similarity flexibly with a kernel function $k : \bbR^\DI \times \bbR^\DI \mapsto \bbR$. $I_n$ is the identity matrix of size $n$. 

Finally, as we do not have access to an online setting where we could directly interrogate users, we set up an offline framework on observed ratings based on estimated feedback values from a backbone model $\FREWARD : \bbR^\DI \times \bbR^\DU \rightarrow \bbR$, where $\FREWARD(\ITEM{}, \USER{})$ is the expected feedback for user $\USER{}$ on item $\ITEM{}$. 
In the case of a  feedback model, if $\ITEM{}$ is the $k^\text{th}$ recommended item to user $\USER{t}$ at time $t$, then $y^t_k = \FREWARD(\ITEM{}, \USER{t})$. $\FREWARD$ can be implemented with any model, provided the latter satisfies some rather unrestrictive assumptions (see Section~\ref{sec:ALGO}), and we use different models in our experiments in Section~\ref{sec:experiments}. Figure~\ref{fig:setting} illustrates this setting. 

\begin{figure}[h]
\begin{center}
\includegraphics[width=0.49\textwidth]{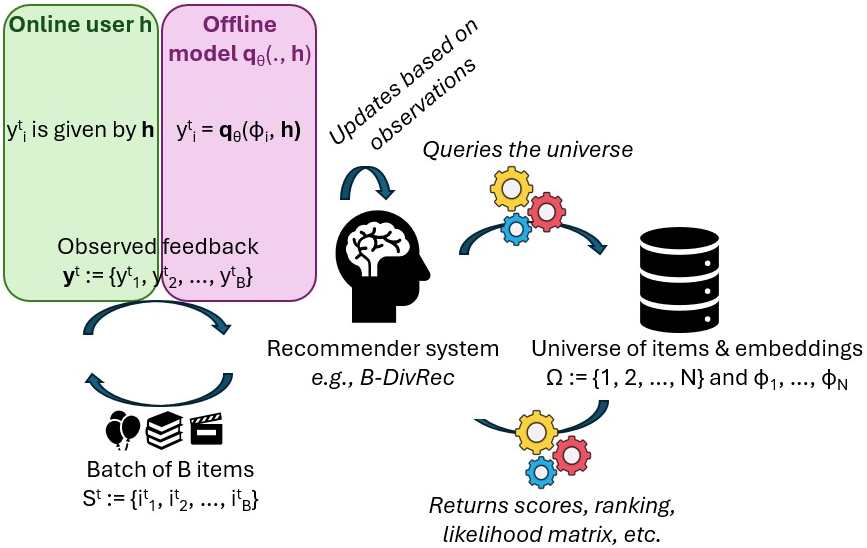}
\end{center}
\caption{Recommendation setting in our paper.}
\label{fig:setting}
\end{figure}

\paragraph{Metrics.} 
To quantify the quality and the diversity of recommendations, we define the following pointwise metrics for any round $t > 0$. 
Quality (also called relevance in this paper and denoted $\text{rel}$) is simply defined by the expected click-through rate. We also consider the precision metric (denoted $\text{prec}$), that is, the ratio of positively rated items (\IE such that the feedback value is higher than a given threshold $\tau>0$) over the total number of recommended items. 
Moreover, we define two types of diversity: the \textit{intrabatch/local diversity} (denoted $\text{div}^\text{L}$), focusing on the diversity inside a batch of $\NBATCH$ items; and the \textit{interbatch/global diversity} (denoted $\text{div}^\text{G}$), which looks at the diversity of the previously and currently recommended items, meaning that it also takes into account the user history. 
An intuitive idea of how diverse a set $\BATCH{}$ of items is can be obtained by computing the volume of the parallelotope induced by the rows of the kernel matrix $K_{\BATCH{},\BATCH{}} := (k(\ITEM{i},\ITEM{j}))_{i,j \in \BATCH{}}$ built from the item embeddings in $\BATCH{}$. The volume of a set $\BATCH{} := \{i_1, \dots, i_{\NBATCH}\}$ of items is $\VOLUME{\BATCH{}} := (\det K_{\BATCH{},\BATCH{}} )^{1/2}$.~\footnote{By definition of a kernel function, this definition of volume is well-defined for any $\BATCH{}\subseteq \Omega$.} At each recommendation round $t$, we define
\begin{eqnarray}\label{eq:def_metrics}
    && \forall t > 0, \ \text{rel}(\BATCH{t}) := \frac{1}{\NBATCH} \sum_{i \in \BATCH{t}} \FREWARD(\ITEM{i}, \USER{t})\:,  \\
    & & \text{div}^\text{L}(\BATCH{t}) := \VOLUME{\BATCH{t}}\:, \text{ and } \text{div}^\text{G}(\BATCH{t}) := \VOLUME{\BATCH{t} \cup \HISTORY{\USER{t}}{t}}\:.  \nonumber
\end{eqnarray}
To assess the performance accrued by a strategy across $\HORIZON$ recommendation rounds, we consider the average of $\text{rel}$, $\text{prec}$, $\text{div}^\text{L}$ and $\text{div}^\text{G}$ across rounds. Finally, we introduce a model- and similarity-agnostic quality-diversity tradeoff metric. 
Given a 
threshold $\tau$, 
the effective diversity $\text{div}^{+}$ contributed by positively-labeled items among recommended batches $\mathbb{S} := \{\BATCH{1},\dots,\BATCH{T}\}$ for 
user $\USER{}$ after $\HORIZON$ rounds is 
\begin{eqnarray}\label{eq:qdt_metric}
\text{div}^{+}(\mathbb{S}, \USER{}) := \VOLUME{\Big\{ i \mid i \in \bigcup_{s < t} \BATCH{s}, \FREWARD(\ITEM{i}, \USER{}) \geq \tau  \Big\}}\:. 
\end{eqnarray}
Albeit the threshold $\tau$ has a large impact on this metric, selecting $\tau$ in real-life settings is usually intuitive. For instance, $\tau$ might be the number of rating stars needed for a good movie in movie recommendation~\citep{harper2015movielens}, whereas it might be the probability of success of a clinical trial in drug repurposing~\cite{liang2017lrssl}.


\paragraph{Contributions.} Any recommender system considered in this paper is applied \textit{directly on the whole library of items}. Contrary to prior works~\citep{ibrahim2025diversified,wang2025beyond} that require several prefiltering by relevance and diversification steps to be tractable, we think that using a single recommendation step on a large library has a better chance of achieving the QDT. It means that anything beyond a time complexity linear in the number of items $\NITEMS$ will be intractable. Our objective is to design a \textit{generic and tractable} recommender system 
implementing the quality-diversity trade-off 
for a possibly large number of items up to millions. First, we introduce an efficient algorithm called \ALGO{} combining determinantal point processes and a fuzzy denuding procedure to adjust the degree of item diversity compared to the user history 
(Section~\ref{sec:ALGO}). Second, we introduce a flexible algorithmic layer to adaptively tune the quality-diversity trade-off to the received user feedback across recommendation rounds, which works for any recommender system in our framework 
(Section~\ref{sec:adaptive_lambda}). Finally, we illustrate the performance and versatility of \ALGO{} on synthetic and public real-life data sets 
(Section~\ref{sec:experiments}). 

\section{A generic approach for the QDT}\label{sec:ALGO}

In this section, we define our trade-off on relevance and diversity, instead of a multi-objective task~\citep{abbassi2009getting} or a multi-step process for filtering relevance or diversity~\citep{yuan2016discovering,ibrahim2025diversified}. We recall that $\KERNEL : \bbR^\DI \times \bbR^\DI \rightarrow \bbR$ is a similarity kernel on items. To use a quality-diversity decomposition, we need:

\begin{assumption}{\textnormal{Positive-definite kernel.}}\label{as:kernel}
    The kernel $\KERNEL$ is a \textnormal{positive-definite kernel}, that is, there exists $d' \in \bbN^*$ and a feature map function $\nu : \bbR^\DI \rightarrow \bbR^{\DI'}$ such that, for any $\bm{x}, \bm{y} \in \Omega$, $\KERNEL(\bm{x}, \bm{y}) = \nu(\bm{x})^\intercal \nu(\bm{y})$. $\nu$ is extended to subsets of items: $\nu(\BATCH{}) := (\nu(\ITEM{i}))_{i \in \BATCH{}} \in \bbR^{\NBATCH \times \DI'}$ for any $\BATCH{} \subseteq \Omega$. 
\end{assumption}

\begin{assumption}{\textnormal{Positive unbounded feedback.}}\label{as:feedback_model}
    The feedback model $\FREWARD$ has values in $\mathbb{R}_+^{*}$. 
\end{assumption}

\begin{assumption}{\textnormal{Unit embeddings.}}\label{as:bounded_norm}
    $\|\ITEM{}\|_2=\|\USER{}\|_2=1$ for all user embeddings $\USER{} \in \bbR^\DU$ and item embeddings $\ITEM{} \in \Omega$.
\end{assumption}

Note that those assumptions are only moderately restrictive, as adequate kernel functions (linear, RBF, Mat\'{e}rn~\citep{williams2006gaussian,duvenaud2014kernel}) and classification or regression models with positive values~\citep{lee1999learning,wood2017generalized} abound in the literature, and renormalizing embeddings is a common procedure. 


\paragraph{The Disentangled Quality-Diversity (DQD) family.} Under these assumptions, we extend the quality-diversity decomposition~\citep{kulesza2010structured} in determinantal point processes to any valid kernel, and we enable the interpolation from only quality-focused recommendation~\footnote{\IE diversity described by the kernel is not considered.} to diversity-focused recommendation~\footnote{\IE we ignore quality scores and only aim at being diverse.}. The likelihood matrix for this family of DPPs for a given user $\USER{} \in \bbR^\DU$, their history $\HISTORY{\USER{}}{}$, and item subset $\BATCH{} \subseteq \Omega$, is:
\begin{eqnarray}\label{eq:extension_qdd}
    \USERKERNEL{\lambda}{f}{\BATCH{}}{\USER{}} := (\MREWARD{\USER{},\BATCH{}})^{2\lambda} f(\KERNEL, \BATCH{}, \HISTORY{\USER{}}{})^{2(1-\lambda)} (\MREWARD{\USER{},\BATCH{}})^{2\lambda} \:,
\end{eqnarray}
where $\lambda \in [0,1]$ controls the trade-off between quality ($\lambda=1$) and diversity ($\lambda=0$). The presence of $\lambda$ allows us to implement continuously and explicity the trade-off between relevance and diversity. $\MREWARD{\USER{},S}$ is a diagonal matrix containing the expected rewards $\{\FREWARD(\ITEM{i}, \USER{})\}$ for $i \in \BATCH{}$. $f$ is a function with values in the set of positive-definite matrices in $\bbR^{\NBATCH \times \NBATCH}$, and depends on kernels computed on $S$, $\Omega$, $\HISTORY{\USER{}}{}$. $f$ should incorporate all information about the desired diversity in recommended items. Note that we completely disentangle quality and diversity in this definition. In our experiments, we consider a linear kernel function. 

Equation~\ref{eq:extension_qdd} offers a flexible definition of the quality-diversity trade-off, which recovers many well-known DPPs 
depending on the definition of 
$f$ and $\lambda$ (see Appendix~\ref{app:related_works}). We aim at maximizing the log-determinant of $\USERKERNEL{\lambda}{f}{\BATCH{}}{\USER{}}$ as $\SETSCORE{}^\lambda_{f,\FREWARD}(\BATCH{};\USER{}) := \log \det \USERKERNEL{\lambda}{f}{\BATCH{}}{\USER{}}$, 
with the convention $\SETSCORE{}^\lambda_{f,\FREWARD}(\emptyset):=0$.~\footnote{Provided that the volume is well-defined--that is, $f(\KERNEL, \Omega, \emptyset)$ is positive-definite for any universe of items $\Omega$.} 
The larger the score, the better the subset $\BATCH{}$ for the quality-diversity trade-off. We could stop at this point, and simply consider a (possibly conditional) DPP to obtain diversified recommendations~\citep{affandi2012markov}. Still, as discussed in the introduction, this approach might be expensive for users with long histories. We suggest another more tractable approach named \ALGO{} below, which leverages the existence of kernel-associated feature maps. 

\begin{algorithm}[tb]
   \caption{Recommendation round at $t>0$ with \ALGO{} ($f=f^\text{BDR}$) \textcolor{purple}{and an adaptive Q-D trade-off}.}
    \label{alg:recom_setting}
\begin{algorithmic}
   \STATE {\bfseries Input:} $\USER{t}$: user context, $\HISTORY{\USER{t}}{t}$: user history, $\FREWARD$: feedback model, \textcolor{purple}{$\mathcal{A}$: online learner, $\lambda^{t}$: initial relevance weight}
   \STATE Initialize $\MREWARD{\USER{t},\Omega} \gets \text{diag}(\{\FREWARD(\ITEM{}, \USER{t})\}_{\ITEM{} \in \Omega})$ 
   \STATE \textcolor{gray}{\# Compute the likelihood matrix}
   \STATE $\USERKERNEL{\textcolor{purple}{\lambda^{t}}}{f}{\Omega}{\USER{t}} \gets (\MREWARD{\USER{t},\Omega})^{2\textcolor{purple}{\lambda^{t}}} f(\KERNEL, \Omega, \HISTORY{\USER{t}}{t})^{2(1-\textcolor{purple}{\lambda^{t}})} (\MREWARD{\USER{t},\Omega})^{2\textcolor{purple}{\lambda^{t}}}$ 
   \STATE \textcolor{gray}{\# Use \textsc{SAMPLING} or \textsc{MAXIMIZATION}}
   \STATE Sample $\BATCH{t}$ with $\mathbb{P}(\BATCH{t}) \propto \det \USERKERNEL{\textcolor{purple}{\lambda^{t}}}{f}{\BATCH{t}}{\USER{t}}$ or Compute $\BATCH{t} \in \arg\max_{\BATCH{} \subseteq \Omega} \det \USERKERNEL{\textcolor{purple}{\lambda^{t}}}{f}{\BATCH{}}{\USER{t}}$
   \STATE Receive $\bm{y}^t \leftarrow (y^t_1, y^t_2, \dots, y^t_\NBATCH)$ from model $\FREWARD$
  \STATE \textcolor{purple}{Update $\mathcal{A} \gets (1-2\lambda^t)\nabla_{\lambda} \SETSCORE{t}^{\lambda^{t}}_{f,\bm{y}^t}(\BATCH{t};\USER{t})$}
  \STATE \textcolor{purple}{Predict $\lambda^{t+1}$ from $\mathcal{A}$}
   \STATE {\bfseries Return} $\BATCH{t}$ (item batch)\textcolor{purple}{, $\lambda^{t+1}$ (updated weight)}
\end{algorithmic}
\end{algorithm}

\paragraph{The \ALGO{} (BDR) DPP.} \ALGO{} belongs to the DQD family, and is defined with  $\alpha \in [0,2]$. If $G^{\BATCH{}, \HISTORY{\USER{}}{}} \in \mathbb{R}^{\NBATCH \times \DI}$, and, for $i \leq \NBATCH$, $G^{\BATCH{}, \HISTORY{\USER{}}{}}_{i,\cdot} = \ITEM{\ell_i}$ if $\underset{j \in \HISTORY{\USER{}}{}} {\max}\cos(\nu(\ITEM{\ell_i}),\nu(\ITEM{j})) \geq 1-\alpha$ (that is, there exists an item in the history $1-\alpha$-close to the considered item with the cosine similarity) and otherwise $\bm{0}$, then
\begin{eqnarray}\label{eq:ALGO}
   f^\text{BDR}(\KERNEL, \BATCH{}, \HISTORY{\USER{}}{}) :=  XX^\intercal\text{ with }X := \nu(\BATCH{}) - \nu(G^{\BATCH{}, \HISTORY{\USER{}}{}})\;.
\end{eqnarray}
$\alpha$ controls in a more subtle way than the kernel $\KERNEL$ the degree of diversity expected compared to the user history, by filtering out items too similar to the history. We discuss the choice of hyperparameter $\alpha$ in Section~\ref{sec:experiments}. See Algorithm~\ref{alg:recom_setting} for a pseudo-code of \ALGO{} with a fixed value of $\lambda \in [0,1]$. \citet{gartrell2017low,dupuy2018learning} also propose a family of low-rank factorizations of $L$-matrices. However, \citet{gartrell2017low} does not endorse a quality-diversity trade-off, whereas \citet{dupuy2018learning} requires two supplementary low-rank approximations. Conversely, \ALGO{} only resorts to 
a single Nystr\"{o}m approximation. 

\paragraph{Scalability of \ALGO.} Naively, the computation of the $L$-matrix in \ALGO{} is in $\mathcal{O}(\NITEMS^3)$. However, we leverage several approximations and methods to achieve a time complexity linear in $\NITEMS$ that we describe more thoroughly in Appendix~\ref{app:tractable}. In practice, we learn function $\nu$ with a Nystr\"{o}m approximation~\citep{nystrom1930praktische,yang2012nystrom,liu2021random} of low rank $\DI'$ to avoid computations on potentially large matrices of size $\NITEMS \times \NITEMS$. This approximation can be applied to any kernel function $\KERNEL$. However, one should be careful to select $\DI'$ such that $\NITEMS \gg \DI' \geq \NBATCH + \max_{t \leq \HORIZON} |\HISTORY{\USER{t}}{t}|$ to ensure that the interbatch/global volume of any sampled batch of $\NBATCH$ items can be larger than $0$ (see Equation~\ref{eq:def_metrics}). 

The Nystr\"{o}m approximation of rank $\DI'$ on a random selection of representative points, run once, has a time complexity of $\mathcal{O}(\NITEMS (\DI')^2 + (\DI')^3)$~\citep{williams2000using}. The computation of $G^{\BATCH{}, \HISTORY{\USER{}}{}}$ for any user $\USER{}$ of history $\HISTORY{\USER{}}{}$ and subset $\BATCH{}$ using a $k$-d tree leads to an average time complexity in $\Omega(\NBATCH \log |\HISTORY{\USER{}}{}|)$, where $|\HISTORY{\USER{}}{}| \ll \NITEMS$ as a general rule, and $\Omega(\NBATCH |\HISTORY{\USER{}}{}|)$ in the worst case of unbalanced trees~\citep{arya1993approximate}. Retrieving the closest neighbor can also be computed on large sets of (feature maps of) items by considering an approximate nearest neighbor algorithm, \EG FAISS~\citep{douze2024faiss} or LSH~\citep{dasgupta2011fast}. The computation of the matrix power in $\lambda$ can also be in linear time in $\NITEMS$. Using the $\alpha$-DPP sampling procedure~\citep{calandriello2020sampling} for the \textsc{SAMPLING} strategy and the greedy algorithm in~\citet{chen2018fast} for the \textsc{MAXIMIZATION} approach, with time complexities linear in $\NITEMS$ (see Appendix~\ref{app:related_works}), we confirm that \ALGO{} remains tractable even when facing millions of items. 

\section{User-adaptive quality-diversity trade-off}\label{sec:adaptive_lambda} 

A frequent problem when dealing with the quality-diversity trade-off is to handle the fact that some users might be more receptive to diverse recommendations than others. 
We wish for an automated procedure for tuning the trade-off which takes into account the user's prior reactions to diversified items. In practice, this boils down to changing the value of hyperparameter $\lambda$. Note that this setting goes beyond regular hyperparameter tuning. First, this occurs in a potentially adversarial context, where the user might switch behavior across recommendation rounds (conservative search versus openness to change), whereas hyperparameter tuning expects some stationarity. Second, this tuning should occur online, and not on an offline training data set, which is often not accessible for new users.

We can frame this problem as an online learning game between a player--that is, the recommender system--and Nature--which is the interaction with \textit{one} user $(\USER{}, \HISTORY{\USER{}}{t})$ querying the recommender system at time $t$. 
At round $t \leq \HORIZON$ of the game, the recommender system chooses a value of $\lambda^t \in [0,1]$. Then, the recommender system makes a recommendation $\BATCH{t}$ to user $\USER{}$ and receives a loss 
value related to the quality-diversity trade-off achieved, depending on the vector of feedback $\bm{y}^t$ from Nature. The recommender system should use this loss to update and use $\lambda^{t+1} \in [0,1]$ at the next round. The game ends after $\HORIZON$ interactions with the same user. The quality-diversity trade-off gain is given by the function $\lambda, \BATCH{}, \bm{y} \mapsto \SETSCORE{t}^\lambda_{f,\bm{y}}(\BATCH{};\USER{})$ for all $\lambda \in [0,1]$, $\BATCH{} \subseteq \Omega$, $|\BATCH{}|=\NBATCH$, and $\bm{y} \in (\bbR_+^{*})^\NBATCH$.~\footnote{Analoguously to a feedback model, $\bm{y}$ is such that $\bm{y}(\ITEM{}, \USER{})$ is the (observed) feedback from the user $\USER{}$ if $\ITEM{}$ has been recommended, and otherwise is equal to $1$ (meaning that it is ignored for a non-visited $\ITEM{}$).} The goal of the game, by selecting the $\lambda_t$'s, is to maximize the cumulated quality-diversity trade-off, and alternatively, to minimize the cumulative regret compared to a deterministic oracle knowing the batch-feedback pairs $(\BATCH{t},\bm{y}^t)_{t \leq \HORIZON}$ in advance
\begin{equation}\label{eq:regret_adaptive}
   \mathcal{R}(\HORIZON) :=  \max_{\lambda \in [0,1]} \sum_{t \leq \HORIZON} \SETSCORE{t}^\lambda_{f,\bm{y}^t}(\BATCH{t};\USER{}) -  \SETSCORE{t}^{\lambda^t}_{f,\bm{y}^t}(\BATCH{t};\USER{})\:.
\end{equation}

Many online learners have been introduced to solve this type of problem, \EG EXP3~\citep{auer2002nonstochastic} and AdaHedge~\citep{de2014follow}. We select AdaHedge, as it is an online learner which is horizon $\HORIZON$-agnostic, and does not require to know the scale of the loss function in advance. We further go into details as regards the implementation of this procedure in Appendix~\ref{subapp:adaptive_implementation}. Algorithm~\ref{alg:recom_setting} (purple lines) shows how we modify the initial recommendation algorithm with \ALGO{} to adaptively choose $\lambda$.   We derive guarantees on the regret incurred by this procedure.
\begin{theorem}{\textnormal{Upper bound on the regret incurred by the adaptive diversity tuning procedure.}}\label{th:regret_UB} An upper bound on the regret $\mathcal{R}(\HORIZON)$ incurred by the adaptive strategy for tuning the level of diversity $\lambda \in [0,1]$ for user $\USER{}$ over $\HORIZON$ rounds of recommendations is 
\begin{eqnarray*}
& & \mathcal{R}(\HORIZON) \leq 2\delta_T\sqrt{T\log(2)}+16\delta_T(2+\log(2)/3)\:,\\
& & \text{ where } \delta_T := 8 \max_{t \leq T} \log \frac{(\max_{i \leq B} y^t_i)^B}{\VOLUME{f(k,S^t,H^t)}}\;.
\end{eqnarray*}
\end{theorem}
The upper bound has a time dependence in $\mathcal{O}(\sqrt{\HORIZON})$, which is on par with the state-of-the-art for online learners.
\begin{proof}
The full proof is shown in Appendix~\ref{subapp:adaptive_bounds}, and we describe here a sketch. We show that the loss function used with AdaHedge at time $t$ in Line 8 of Algorithm~\ref{alg:recom_setting} is proportional to $\lambda$. Thus we show that the loss function is linear (and then, concave) in $\lambda$. Applying successively the gradient trick, Theorem 8 and Corollary 17 from~\cite{de2014follow} to the convex function $-g_t$ and using standard upper bounds leads to the regret upper bound.
\end{proof}

\section{Experimental study}\label{sec:experiments}

To try to get as close as possible to a realistic online setting, we consider the following situation. At recommendation time $t>0$, for a new user $\USER{t}$ with ground-truth history $\HISTORYGT{\USER{t}}{t}$ from the initial data set, a recommender system will output a batch of recommendations for each round $t,t+1,\dots,t+|\HISTORYGT{\USER{t}}{t}|$ with respective user context-history pairs $(\USER{t}, \emptyset), (\USER{t}, \{ i_1 \}), (\USER{t}, \{ i_1, i_2 \}), \dots, (\USER{t}, \HISTORYGT{\USER{t}}{t})$, where $\HISTORYGT{\USER{t}}{t} := \{ i_1, i_2, \dots, i_{M} \}$ if $|\HISTORY{\USER{t}}{t}|=M$ in the initial data set. It means that at each round, we incrementally increase the user history with true previously recommended items. 
We iterate this process over $10$ random seeds, and average all the metrics across those 10 iterations for each user. For pointwise metrics such as in Equation~\ref{eq:def_metrics}, we also average these metrics along the trajectory of recommendations of length $|\HISTORYGT{\USER{t}}{t}|$. 
Finally, since we run this setting on several users, we aggregate these metrics across users by considering the average and the standard deviation. 
The (execution) time is the time in seconds needed to output a single batch of recommendations to a user. All values are rounded to the closest $2^\text{nd}$ decimal place.

We consider the following data sets in our experiments: (1) synthetic sets: 
their names are prefixed with \textsc{SYNTHETIC}, followed by the number of items $\NITEMS$, with feedback values in $[0,1]$. 
; (2) real-life data sets: MovieLens data set~\citep{harper2015movielens} for movie recommendation; 
Epinions data set~\citep{leskovec2010signed} in social networks, both with feedback values in $[0,5]$, and a collection of six data sets for drug repurposing (Cdataset, DNdataset, Fdataset, Gottlieb, LRSSL, and PREDICT) with feedback values in $[0,1]$. 
As previously mentioned, we use different feedback models for all data sets to illustrate the generality of our results. See respectively Appendices~\ref{app:synthetic}-\ref{app:real_life} for the synthetic and real-life data sets.

\textbf{Hyperparameter sensitivity analysis.} We vary values of hyperparameters $\lambda \in [0,1]$ $\alpha \in [0,2]$ (only for \ALGO{}). The trends for relevance and diversity depending on $\lambda$ and $\alpha$ are shown respectively in Figure~\ref{fig:SYNTHETIC750_lambda_alpha} 
and Table~\ref{tab:SYNTHETIC750_lambda_alpha}. 
As expected, across all recommender systems, as $\lambda$ increases, relevance increases, whereas (local or global) diversity globally decreases. 
Visually, the top two recommender systems are our contribution \ALGO{} and MMR. Computing explicitly the areas under the curve for each metric (rounding to the closest third decimal place) shows that \ALGO{} outperforms MMR on all metrics except for the global diversity. 
Regarding $\alpha$ in \ALGO, unsurprisingly, the higher $\alpha$, the more (globally) diverse the recommendations, to the price of a loss in relevance and precision. Beyond a certain value of $\alpha$ (for this data set, $\alpha=1$), the DPP can no longer sample enough diverse items, as $\alpha$ filters out most of the items. We discuss this issue and potential solutions in Appendix~\ref{app:tractable}. In the remainder of the section, we set $\lambda=0.5$ and $\alpha=0$.

\begin{figure}
   \includegraphics[width=0.49\textwidth]{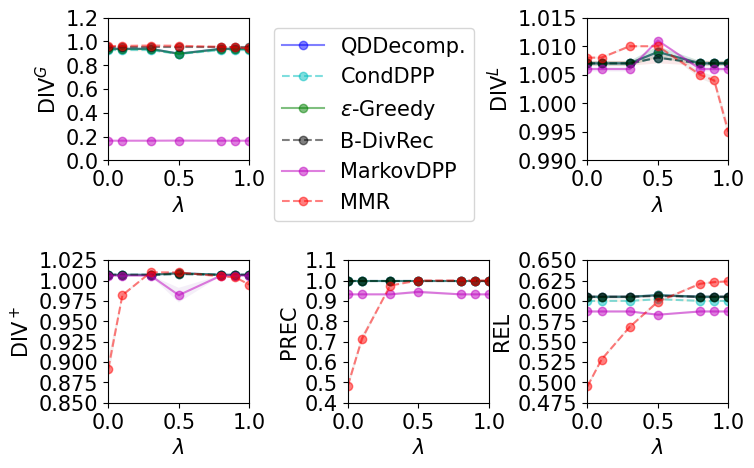}%
 \caption{Sensitivity analysis for $\lambda$ on SYNTHETIC750 (6 users, $\NBATCH=3$, $\tau=0.5$), with the \textsc{MAXIMIZATION} strategy.}%
 \label{fig:SYNTHETIC750_lambda_alpha}
 \end{figure}

 \begin{table}[t]
  \caption{Sensitivity analysis for $\alpha$ on SYNTHETIC750 (6 users, $\NBATCH=3$, $\tau=0.5$), \ALGO{} with the \textsc{MAXIMIZATION} strategy.}
  \label{tab:SYNTHETIC750_lambda_alpha}
  \begin{center}
    \begin{small}
      \begin{sc}
        \begin{tabular}{lcccr}
          \toprule
         $\alpha$  & rel $\uparrow$         & prec $\uparrow$      & div$^\text{G}$ $\uparrow$  \\
          \midrule
        0 &0.61$\pm$ 0.0&1.0$\pm$ 0.0&0.95$\pm$ 0.0\\
        0.1&0.61$\pm$ 0.0&1.0$\pm$ 0.0&0.95$\pm$ 0.0 \\
        0.5&0.61$\pm$ 0.0&1.0$\pm$ 0.0&0.95$\pm$ 0.0 \\
        0.8&0.61$\pm$ 0.0&1.0$\pm$ 0.0&0.96$\pm$ 0.0\\
        0.90&0.6$\pm$ 0.0&1.0$\pm$ 0.0&0.97$\pm$ 0.0\\
        0.95&0.58$\pm$ 0.0&1.0$\pm$ 0.0&0.97$\pm$ 0.0\\
        0.98&0.57$\pm$ 0.0&0.97$\pm$ 0.0&0.97$\pm$ 0.0 \\
        1 &0.56$\pm$ 0.0&0.92$\pm$ 0.01&0.96$\pm$ 0.0\\
          \bottomrule
        \end{tabular}
      \end{sc}
    \end{small}
  \end{center}
  \vskip -0.1in
\end{table}


\paragraph{Testing the scalability of \ALGO{}.} 
We run an experiment with \ALGO{} and its strongest contender MMR on SYNTHETIC1500, with varying batch sizes. We report the average execution time in seconds over 10 iterations in Table~\ref{tab:vary_batch_sizes}. Albeit MMR is slightly faster for smaller values of batches, it becomes quickly much slower as the batch size $\NBATCH$ increases, and up to $12\times$ slower when $\NBATCH=900$.

 \begin{table}[t]
  \caption{Execution time benchmark (runtime in seconds) on SYNTHETIC1500 (4 users, $\tau=0.5$) with \textsc{MAXIMIZATION}.}
  \label{tab:vary_batch_sizes}
  \begin{center}
    \begin{small}
      \begin{sc}
        \begin{tabular}{lccc}
          \toprule
         Batch size $\NBATCH$  & MMR         & \textbf{\ALGO{}}   \\
          \midrule
        5 &\textbf{0.03$\pm$ 0.0}&0.07$\pm$ 0.0\\
        50&0.19$\pm$ 0.0&\textbf{0.08$\pm$ 0.0} \\
        500&2.05$\pm$ 0.01&\textbf{0.21$\pm$ 0.0} \\
        900&4.16$\pm$ 0.13&\textbf{0.32$\pm$ 0.0}\\
          \bottomrule
        \end{tabular}
      \end{sc}
    \end{small}
  \end{center}
  \vskip -0.1in
\end{table}

\paragraph{Noiseless feedback setting.} First, we compared the \textsc{SAMPLING} and \textsc{MAXIMIZATION} strategies in Appendix~\ref{app:compl_experiments} on SYNTHETIC\{750,30k,3M,15M\}, the MovieLens and the PREDICT data sets. See Table~\ref{tab:SYNTHETIC750} for the results on SYNTHETIC750. We remarked that MarkovDPP, which is the baseline which is closest to our setting, could not be run on instances with more than $5,000$ items. Moreover, we consistently found across data sets that the \textsc{MAXIMIZATION} strategy achieved a better quality-diversity tradeoff than the \textsc{SAMPLING} strategy. This motivated us to only use this strategy for all DPP-based approaches in the remainder of the paper.  
 \begin{table}[t]
  \caption{Benchmark on SYNTHETIC750 (6 users, $\NBATCH=3$).}
  \label{tab:SYNTHETIC750}
  \begin{center}
    \begin{small}
      \begin{sc}
        \begin{tabular}{lccc}
          \toprule
        \textsc{SAMPLING} & \text{rel} $\uparrow$ & \text{div}$^G$  $\uparrow$ & \text{div}$^+$ $\uparrow$   \\
          \midrule
        QDDecomp.&0.51$\pm$ 0.0&\underline{0.95$\pm$ 0.00}&\textbf{1.01$\pm$ 0.0}\\ 
        CondDPP&0.51$\pm$ 0.0&0.94$\pm$ 0.01&0.97$\pm$ 0.01 \\ 
        $\varepsilon$-Greedy&0.51$\pm$ 0.0&\underline{0.95$\pm$ 0.00}&\textbf{1.01$\pm$ 0.0} \\ 
        MarkovDPP&0.51$\pm$ 0.0&0.17$\pm$ 0.00&0.91$\pm$ 0.01 \\ 
        \textbf{\ALGO}&0.51$\pm$ 0.0&\underline{0.95$\pm$ 0.00}&\underline{0.98$\pm$ 0.01}\\ 
        \midrule
        \midrule
        \textsc{MAXIMIZ.} & &  &   \\
     \midrule
    QDDecomp.&\textbf{0.61$\pm$ 0.0}&0.90$\pm$ 0.01&\textbf{1.01$\pm$ 0.0}\\ 
    CondDPP&\underline{0.60$\pm$ 0.0}&0.89$\pm$ 0.01&\textbf{1.01$\pm$ 0.0}\\ 
    $\varepsilon$-Greedy&\textbf{0.61$\pm$ 0.0}&0.90$\pm$ 0.01&\textbf{1.01$\pm$ 0.0}\\ 
    MarkovDPP&0.58$\pm$ 0.0&0.17$\pm$ 0.0&\underline{0.98$\pm$ 0.01}\\ 
    \textbf{\ALGO}&\textbf{0.61$\pm$ 0.0}&\underline{0.95$\pm$ 0.0}&\textbf{1.01$\pm$ 0.0}\\ 
        \midrule
        \midrule
        MMR&\underline{0.60$\pm$ 0.0}&\textbf{0.96$\pm$ 0.0}&\textbf{1.01$\pm$ 0.0}\\ 
          \bottomrule
        \end{tabular}
      \end{sc}
    \end{small}
  \end{center}
  \vskip -0.1in
\end{table}
We now focus our study on known baselines Deep DPP~\citep{gartrell2018deep} (where the likelihood matrix is learned by a neural network), MMR~\citep{carbonell1998use} and xQuAD~\citep{santos2010exploiting} (an intent-aware model) instead of DQD variants which are tested in Appendix~\ref{app:compl_experiments}.
Full results on all metrics are located in Appendix~\ref{app:benchmark}.
On the MovieLens data set for movie recommendation, MMR is the top contender, whereas our contribution \ALGO{} is the second best, improving significantly upon all other baselines diversity-wise. The $\text{div}^\text{G}$ values are particularly small on MovieLens due to the fact that initial user histories are collinear: they feature similar movie embeddings, which leads to a small volume, due to having almost the same metadata in the data set. The same phenomenon can be observed on the Epinions data set. We discuss this issue further with a quantitative analysis in Appendix~\ref{app:discussion_movielens}. However, on all other data sets, \ALGO{} achieves the quality-diversity tradeoff, with relevance values close to the quality-wise top performers MMR and xQuAD, while largely improving on the global diversity metric. This is confirmed by the fact that \ALGO{} has the best value of effective diversity $\text{div}^{+}$ across all eight data sets but one. 

\begin{table}[tb!]
  \caption{Benchmark on real-life data sets: Cdataset (4 users, $\NBATCH=3$), DNdataset (4 users, $\NBATCH=3$), Epinions (3 users, $\NBATCH=3$: DeepDPP could not be run on Epinions, as the number of items is too large to compute the likelihood function in memory), Fdataset (4 users, $\NBATCH=3$), LRSSL (4 users, $\NBATCH=3$), MovieLens (4 users, $\NBATCH=3$) and PREDICT (4 users, $\NBATCH=3$).}
  \label{tab:benchmark}
  \begin{center}
    \begin{small}
      \begin{sc}
        \begin{tabular}{lcccrr}
          \toprule
           Cdataset         & \text{rel}  $\uparrow$  & \text{div}$^G$ $\uparrow$ &  \text{div}$^+$ $\uparrow$ \\
          \midrule
             DeepDPP & 0.36$\pm$ 0.03 &\underline{0.16$\pm$ 0.02}  & 0.32$\pm$ 0.04\\
             MMR&   \textbf{0.91$\pm$ 0.01}  &0.14$\pm$ 0.01 &0.36$\pm$ 0.0 \\
             xQuAD & \textbf{0.91$\pm$ 0.01} &0.14$\pm$ 0.01 & \underline{0.37$\pm$ 0.01}\\
             \textbf{\ALGO}& \underline{0.79$\pm$ 0.01} & \textbf{0.22$\pm$ 0.02} &\textbf{0.52$\pm$ 0.02}\\
             \midrule
             DNdataset         &  & &   \\
             \midrule
             DeepDPP & 0.23$\pm$ 0.06 &0.72$\pm$ 0.01 &\underline{0.23$\pm$ 0.06}  \\
             MMR & \textbf{0.33$\pm$ 0.06} &\underline{0.86$\pm$ 0.03} &0.19$\pm$ 0.05  \\
             xQuAD & \textbf{0.33$\pm$ 0.06} &0.77$\pm$ 0.01&0.22$\pm$ 0.06\\
             \textbf{\ALGO} & \underline{0.31$\pm$ 0.06}  & \textbf{0.98$\pm$ 0.01} &\textbf{0.25$\pm$ 0.07} \\
             \midrule
             Epinions         &  & &   \\
             \midrule
            Deep DPP & -- & -- & --\\
            MMR&     \textbf{0.04$\pm$ 0.01}   &0.0$\pm$ 0.0  &0.0$\pm$ 0.0 \\
            xQuAD & \textbf{0.04$\pm$ 0.01}  &0.0$\pm$ 0.0 & 0.0$\pm$ 0.0\\
            \textbf{\ALGO}& \underline{0.02$\pm$ 0.01} & 0.0$\pm$ 0.0 & 0.0$\pm$ 0.0\\
             \midrule
             Fdataset         &  & &   \\
             \midrule
          Deep DPP &  0.48$\pm$ 0.03  & 0.12$\pm$ 0.01  & \underline{0.59$\pm$ 0.06} \\
            MMR&  \textbf{0.95$\pm$ 0.01 } &\underline{0.16$\pm$ 0.01 }&0.35$\pm$ 0.01  \\
           xQuAD &   \underline{0.93$\pm$ 0.01} &0.14$\pm$ 0.01&0.37$\pm$ 0.01\\
          \textbf{\ALGO}&   0.81$\pm$ 0.01  & \textbf{0.24$\pm$ 0.02} &\textbf{0.64$\pm$ 0.05} \\
             \midrule
             Gottlieb         &  & &   \\
             \midrule
            Deep DPP & 0.50$\pm$ 0.03 & 0.19$\pm$ 0.02 & 0.55$\pm$ 0.05 \\
            MMR&   \textbf{0.94$\pm$ 0.01} &\underline{0.37$\pm$ 0.03}&0.64$\pm$ 0.02 \\
            xQuAD & \underline{0.93$\pm$ 0.01} &0.36$\pm$ 0.03&\underline{0.65$\pm$ 0.01} \\
            \textbf{\ALGO}&0.82$\pm$ 0.01 & \textbf{0.40$\pm$ 0.02} & \textbf{0.78$\pm$ 0.01 }\\
             \midrule
             LRSSL         &  & &   \\
             \midrule
          Deep DPP &  0.39$\pm$ 0.02  & 0.53$\pm$ 0.05 & 0.64$\pm$ 0.03 \\
          MMR &  \textbf{0.97$\pm$ 0.0} &\underline{0.55$\pm$ 0.05} &\underline{0.82$\pm$ 0.02}\\
         xQuAD &   \textbf{0.97$\pm$ 0.0} &0.53$\pm$ 0.04&0.80$\pm$ 0.02\\
          \textbf{\ALGO}&  \underline{0.83$\pm$ 0.01}  & \textbf{0.56$\pm$ 0.04} & \textbf{0.86$\pm$ 0.01} \\
             \midrule
             MovieLens         &  & &   \\
             \midrule
            Deep DPP & 1.25$\pm$ 0.03 & 0.02$\pm$ 0.00 &  0.22$\pm$ 0.02 \\
            MMR&     \underline{  3.78$\pm$ 0.08} & \textbf{0.06$\pm$ 0.01} & \textbf{0.77$\pm$ 0.03}\\
            xQuAD &  \textbf{ 3.85$\pm$ 0.06} & \underline{0.05$\pm$ 0.00} & \underline{0.68$\pm$ 0.02}\\
            \textbf{\ALGO}& 3.49$\pm$ 0.09  & \textbf{0.06$\pm$ 0.01} &  0.60$\pm$ 0.03 \\
             \midrule
            PREDICT        &  & &   \\
             \midrule
           DeepDPP &  0.21$\pm$ 0.02&0.35$\pm$ 0.03&0.30$\pm$ 0.04\\
           MMR & 0.66$\pm$ 0.02  & \textbf{0.60$\pm$ 0.02} & \underline{0.70$\pm$ 0.02} \\ 
           xQuAD & \textbf{0.79$\pm$ 0.02}  & 0.47$\pm$ 0.03 & 0.64$\pm$ 0.02 \\
          \textbf{\ALGO}&  \underline{0.76$\pm$ 0.02}&\underline{0.57$\pm$ 0.02}&\textbf{1.02$\pm$ 0.02}\\ 
          \bottomrule
        \end{tabular}
      \end{sc}
    \end{small}
  \end{center}
  \vskip -0.1in
\end{table}

\paragraph{Adaptive quality-diversity trade-off.} Finally, we evaluate our adaptive procedure for tuning the quality-diversity trade-off parameter $\lambda \in [0,1]$. Note that $\lambda$ is specific to one user: each online learner is initialized and updated along a trajectory corresponding to a single user. As such, we only consider one user, user with identifier 0, in all data sets. Then, we run \ALGO{} (Equation~\ref{eq:ALGO}) combined with the \textsc{MAXIMIZATION} strategy and the adaptive approach described in Section~\ref{sec:adaptive_lambda}. We report the results in Table~\ref{tab:adaptive} for SYNTHETIC750, 
MovieLens and PREDICT. To assess the goodness of our approach, the final tuned value $\lambda_\text{final}$ is compared to the best \textit{a posteriori} relevance weight $\lambda_\star$. This value is computed at the end of a trajectory associated with one user by solving the maximization problem in $\lambda$ in Equation~\ref{eq:regret_adaptive}, and corresponds to the best deterministic action that the player could have taken in the game with Nature. 

We first focus on the noiseless setting. First, the adaptive procedure seems to be able to roughly retrieve the oracle value $\lambda_\star$ across data sets--albeit it is unlikely that it can always find it, since the oracle value relies on the \textit{a posteriori} knowledge of the user feedback. Second, the computational cost of using the adaptive procedure instead of a fixed value of $\lambda=0.5$ is moderate: the runtime with the adaptive procedure is multiplied by $8$ in average across all data sets, due to the approximation of matrix power (see Appendix~\ref{app:tractable}). Third, on both 
MovieLens and PREDICT
, adaptively selecting $\lambda$ throughout the trajectory allows us to noticeably increase relevance while trading off some diversity. 
MovieLens, 
since the users seem to be more biased towards popular movies compared to novel recommendations--as illustrated by the low value of $\text{div}^\text{G}$
--the recommender system is leaning toward making popular recommendations, leading to a higher $\lambda$, higher relevance ($\text{rel}$ $+2.5\%$) to the price of some of effective diversity ($\text{div}^{+}$ $-1.3\%$). 

\paragraph{Noisy feedback setting.} 
In the synthetic data set, the recommender system no longer receives the exact feedback score $\FREWARD(\ITEM{i},\USER{t})$ at time $t$, but a binary outcome in $\{1,2\}$ determined by a Bernouilli law of probability $p := \min(1, \max(0, \FREWARD(\ITEM{i}, \USER{t})))$. Note that 1 is drawn with probability $1-p$, and 2 with probability $p$. This setting simulates the case where we only receive clicks, and no longer probabilities of clicks. We also add noise to the drug repurposing task (PREDICT data set), where the recommender system has only access to the outcomes of clinical trials (3: successful, 2: not tested, 1: unsuccessful) instead of probabilities of success; and to the movie recommendation task, where observed feedback is the movie rating from 1 (not seen) to 6 (five stars) from the MovieLens data set. We need positive observed feedback to satisfy Assumption~\ref{as:feedback_model}. However, we renormalize ratings in Table~\ref{tab:adaptive} to enhance comparability to the noiseless setting.


As a general rule, considering noisy feedback expectedly decreases the relevance values, but this effect can be mitigated when using the adaptive procedure, which is another argument in favor of this approach in real-life settings. These experiments validate our findings even when we only have access to noisy observable feedback. 
The adaptive procedure allows us to improve on the relevance, when possible. When the relevance cannot be improved further, the relevance-wise and diversity-wise performance remains the same as in the non-adaptive setting.


\begin{table}[t]
  \caption{Benchmark with and without (``standard'') the adaptive tuning procedure in Section~\ref{sec:adaptive_lambda}, applied to \ALGO{} (Equation~\ref{eq:ALGO}) with the \textsc{MAXIMIZATION} strategy for user 0, starting with initial value $\lambda_0=0.5$. 
  We also test the adaptive procedure with a noisy feedback (``with noise''), where $y^t_k \neq \FREWARD(\ITEM{}, \USER{t})$ if $\ITEM{}$ is the k$^\text{th}$ recommended item. 
  We compare metrics in the adaptive setting to the standard setting: bold type denotes similar or improved performance.}
  \label{tab:adaptive}
  \begin{center}
    \begin{small}
      \begin{sc}
        \begin{tabular}{lrrr}
          \toprule
         Standard  & Synth750        & MovieLens      & PREDICT  \\
          \midrule
          \text{rel} $\uparrow$   & 0.61$\pm$ 0.0 & 4.38$\pm$ 0.0&  0.79$\pm$ 0.01\\
          \text{div}$^\text{G}$ $\uparrow$  & 0.96$\pm$ 0.0 & 0.01$\pm$ 0.0&  0.49$\pm$ 0.0 \\
          \text{div}$^+$ $\uparrow$   &  1.01$\pm$ 0.0 & 0.78$\pm$ 0.0 & 0.72$\pm$ 0.01 \\
          Time $\downarrow$    & 0.16$\pm$ 0.0  & 29.7$\pm$ 0.82& 0.82$\pm$ 0.0 \\
          \midrule
          Adaptive &   &    &   \\
          \midrule
          $\lambda_\text{final}$    & 0.18 & 1.00 &  0.35 \\
          $\lambda_\star$ & 0.10 &1.00 & 0.10 \\
          \text{rel} $\uparrow$   & \textbf{0.61$\pm$ 0.0}  & \textbf{4.49$\pm$ 0.01} & \textbf{0.79$\pm$ 0.01} \\
          \text{div}$^\text{G}$ $\uparrow$  & 0.95$\pm$ 0.0 & \textbf{0.01$\pm$ 0.0} & \textbf{0.49$\pm$ 0.0} \\
          \text{div}$^+$ $\uparrow$   & \textbf{1.01$\pm$ 0.0} &0.77$\pm$ 0.01 & \textbf{0.72$\pm$ 0.01}  \\
          Time $\downarrow$    & 0.5$\pm$ 0.0 & 158.1$\pm$ 0.84 &  1.97$\pm$ 0.05\\
          \midrule
          \midrule
          Standard &  with   &  noise  &   \\
          \midrule
          \text{rel} $\uparrow$   & 1.00$\pm$ 0.0 & 2.57$\pm$ 0.26 &  0.28$\pm$ 0.05 \\
          \text{div}$^\text{G}$ $\uparrow$  & 0.16$\pm$ 0.0 & 0.02$\pm$ 0.0 & 0.33$\pm$ 0.03  \\
          \text{div}$^+$ $\uparrow$   & 0.07$\pm$ 0.02 & 0.94$\pm$ 0.0 &  0.83$\pm$ 0.01\\
          Time $\downarrow$    &  0.06$\pm$ 0.0  &2.72$\pm$ 0.01  &1.36$\pm$ 0.08  \\
          \midrule
          Adaptive &  with   &  noise  &   \\
          \midrule
          $\lambda_\text{final}$    & 0.73  & 1.00&   0.97\\
          $\lambda_\star$ &  0.5 &1.00 & 0.99 \\
          \text{rel} $\uparrow$   & \textbf{1.00$\pm$ 0.0} & \textbf{5.00$\pm$ 0.0} &  \textbf{0.41$\pm$ 0.04}\\
          \text{div}$^\text{G}$ $\uparrow$  & \textbf{0.16$\pm$ 0.0} & \textbf{0.02$\pm$ 0.0} &  0.21$\pm$ 0.03 \\
          \text{div}$^+$ $\uparrow$   & \textbf{0.07$\pm$ 0.02} & \textbf{0.94$\pm$ 0.01} & 0.54$\pm$ 0.02 \\
          Time $\downarrow$    &   1.50$\pm$ 0.11 & 57.56$\pm$ 15.55 & 7.7$\pm$ 0.17  \\
          \bottomrule
        \end{tabular}
      \end{sc}
    \end{small}
  \end{center}
  \vskip -0.1in
\end{table}

\section{Discussion}\label{sec:discussion}

In this paper, we first introduced a flexible framework for diverse and good recommendation, as a generic family of DPPs. Then, 
we introduced \ALGO{} 
to integrate 
user history and achieve the QDT in a computationally tractable fashion. Finally, we proposed an algorithmic layer for any member of the DQD family to adaptively tune the level of diversity to the user across recommendation rounds. 

Our work still suffers from some limitations. First, as mentioned in Section~\ref{sec:experiments}, beyond a certain value of $\alpha$, \ALGO{} can no longer sample enough diverse items, as $\alpha$ filters out most of the items. This issue is not specific to our contribution, as \citet{ibrahim2025diversified} has also reported it for other conditional DPPs. It arises when the rank of the likelihood matrix ($L$-matrix) of the DPP is strictly smaller than $\NBATCH$. Possible solutions to bypass this issue would be (1) to return a batch of recommended items of size smaller than $\NBATCH$ (with a risk of no longer recommending any item at some point), (2) to only focus on quality and set the matrix $f(\KERNEL, \BATCH{}, \HISTORY{\USER{}}{}) \gets I_\NITEMS$, (3) to forget some of the user history, and only preserve either the most representative points or the latest recommendations. Second, when the features are not informative enough, the diversity metrics might collapse quickly due to the existence of very similar items in the history. We discuss this issue in depth in Appendix~\ref{app:discussion_movielens}, and it could be solved with (3). Future work on both topics would be of interest for practical recommender systems.

\section*{Impact Statement}

This paper presents work whose goal is to advance the field of machine learning. There are many potential societal consequences of our work, none of which we feel must be specifically highlighted here.

\bibliography{icml2026_conference}
\bibliographystyle{icml2026}

\clearpage
\appendix
\onecolumn

\section{Related works (continued)}\label{app:related_works}

Due to space constraints in the main text, we write here a more comprehensive section than in Section~\ref{sec:related_works} about prior works on recommender systems with diversity. 

\subsection{Diversity metrics}

In Section~\ref{sec:introduction}, we focused on a geometrically intuitive definition of diversity relying on the volume of the parallelotope built by the columns of the kernel function applied to a given subset of items. That definition is versatile enough that, not only could we use it to define intrabatch/local and interbatch/global diversity metrics~\citep{bederina2025bayesian} (see Equation~\ref{eq:def_metrics}), but it also enabled us to produce a single metric for the quality-diversity trade-off (that is, the $\text{div}^{+}$ metric in Equation~\ref{eq:qdt_metric}). These metrics hold for any feedback model and kernel function, hence they are very flexible. As described in Section~\ref{sec:related_works}, prior works~\citep{carbonell1998use,kaminskas2016diversity} also proposed a linear or convex combination of relevance and diversity metrics, controlled by a parameter $\lambda \leq 0$, or a quality-diversity trade-off induced by a matrix decomposition~\citep{kulesza2010structured}. In our paper, we actually combine the two approaches. We discuss below other metrics mentioned in the literature. Note that we make a distinction between metrics which are only evaluated in the experimental studies (\EG intralist average distance or category coverage~\citep{ge2010beyond,kaminskas2016diversity,chen2018fast,li2024contextual}) from the metrics which are optimized upon and actually used in the recommender systems. We focus on the latter in this section.

A strong competitor of our metrics are ridge leverage scores. Given the ridge hyperparameter $\zeta \geq 0$, for a given kernel function $\KERNEL$, \citet[Definition $1$]{musco2017recursive} define the ridge leverage score of item $\ITEM{i} \in \Omega := \{\ITEM{1},\ITEM{2},\dots,\ITEM{\NITEMS}\}$ as $ r^\zeta_{\KERNEL}(\ITEM{i}) := (L(L+\zeta I_\NITEMS)^{-1})_{i,i}$, where $L := \KERNEL_{\Omega,\Omega}$. Alternatively, defining $U \in \bbR^{\NITEMS \times \DI}$ such that $K = UU^\intercal$, the ridge leverage score for $\ITEM{i}$ is the value of the ridge regression problem $\min_{\bm{y} \in \bbR^{\NITEMS}} \|U_{i,\cdot} - \bm{y} U\|_2^2 + \zeta\|\bm{y}\|^2_2 \in [0,1].$ In other words, $\bm{y}$ is the mixing vector that achieves the linear combination of rows in $U$ closest to the $i^\text{th}$ row of $U$ (which is the row associated with $\ITEM{i}$). Then the ridge leverage score is a measure of the unicity of the vector $\ITEM{i}$ among elements in $\Omega$: the smaller it is, the least ``unique'' $\ITEM{i}$ is. 
However, computing the exact ridge leverage scores requires potentially inverting a matrix of size $\NITEMS \times \NITEMS$ which itself has a time complexity of $\mathcal{O}(\NITEMS^3)$. To get a measure of diversity for a whole set of items $\BATCH{}$, summing all leverage scores across $\BATCH{}$ yields the ``effective dimension'', also called ``degrees of freedom'': $d^\zeta_{\text{eff},\KERNEL}(\BATCH{}) := \sum_{\ITEM{i} \in \BATCH{}} r^\zeta_{\KERNEL}(\ITEM{i}) = \text{tr}(L(L+\zeta I_\NITEMS)^{-1})\:.$ 
This measure of diversity is quite intuitive, and the time complexity is in $\mathcal{\widetilde{O}}(\NITEMS)$, where $\mathcal{\widetilde{O}}$ hides logarithmic terms in $\NITEMS$~\citep{chen2021fast}. Then, similarly to determinants and volumes where Cholesky decompositions on sparse or low-rank matrices can be leveraged (see Appendix~\ref{app:tractable}), the time complexity of these measures might be linear in the number of items. However, since the determinant is directly optimized (in the \textsc{MAXIMIZATION} strategy) or considered (in the \textsc{SAMPLING} strategy) in our algorithms, an evaluation of recommender systems by the volume might be more consistent than by leverage scores.

As mentioned in Section~\ref{sec:introduction}, the topic of the quality-diversity trade-off is related to serendipity--which is also called out-of-the-box~\citep{abbassi2009getting}, unexpectedness~\citep{xu2024serendipitous}, surprise~\citep{ziarani2021serendipity}--where an unexpectedly good item is recommended to a user, contrary to simply recommending items which are very different from those in the user history or cover different categories (which would be diversity). Looking for serendipitous items allows us to overcome the popularity bias in recommendation~\citep{yu2024long}. \citet{kaminskas2016diversity,poulain2020investigating,ziarani2021serendipity,kotkov2024dark} show that there is no consensus regarding the concept of serendipity, nor a single metric. The same holds for diversity metrics: recently, \citet{mironov2024measuring} tried to formalize the expected properties of a good diversity metric, showed that known diversity metrics do not match them, and proposed two matching metrics. These properties are monotonicity (the diversity of the union of two sets should be higher than the maximum diversity for each of these sets), uniqueness (replacing an element from a set with a copy of an already present element should decrease the diversity of the modified set), and continuity of the diversity metric. In particular, the determinant does not satisfy the monotonicity property. Yet the metrics that \citet{mironov2024measuring} proposed are too computationally expensive in practice. In~\citet{yu2024long}, authors built a recommender system for long-tail Web services/items likely to be queried by applications/users, which is based on linear propagation of information on graphs of interconnected services with a LightGCN~\citep{he2020lightgcn}. This allows them to learn item and user embeddings. The predicted probability of querying a service by an application is the inner product of the corresponding embeddings. A recommender system is then trained on those predicted scores with a pairwise Bayesian Personalized Ranking (BPR) loss, which is appropriate with implicit and sparse feedback data. Other papers also study diversity through random walks on graphs~\citep{poulain2020investigating}. Contrary to our paper, most of the reported metrics combine metrics of item unpopularity (related to relevance) and item dissimilarity~\citep{iaquinta2008introducing,abbassi2009getting} to be applied at a reranking stage where a subset of relevant items has already been retrieved. However, as reported in~\citet{li2024contextual}, filtering too early for relevant items might cut down diversity in an irreversible way. Moreover, albeit the presence of numerous novelty-related metrics, the selection of a single metric is subjective, whereas the quality and diversity metrics are usually straightforwardly guided by the recommendation task. 

\subsection{About distributions on random subsets of points}

This section focuses on theoretical developments in DPPs which are relevant to our paper.

Sampling algorithms for DPPs are naively in $\mathcal{O}(\NITEMS^3)$~\citep[Theorem 7 and Algorithm 18]{hough2006determinantal} with $\NITEMS$ the number of available items, but their complexity may be reduced down to $\mathcal{O}(\alpha \NITEMS \cdot \text{poly}(\NBATCH))$ when sampling a subset of fixed size $\NBATCH$ with an $\alpha$-DPP~\citep{calandriello2020sampling} where $\alpha \leq 1$. However, one might want to find the subset with highest probability--that is, the most diverse--instead of sampling according to the DPP. The associated maximization problem is NP-hard~\citep{ko1995exact,grigorescu2022hardness}, but there are greedy approximations that we discuss in this paragraph. 
All in all, selecting $\NBATCH$ elements among $\NITEMS$ using a determinantal point process (DPP) has a time complexity of $\mathcal{O}(\NITEMS \NBATCH^2)$, after finding the eigenvalues of the likelihood matrix, which has a time complexity of $\mathcal{O}(\NITEMS^3)$. In the case of a linear kernel on $\DI$-dimensional item embeddings, finding the eigenvalues has a time complexity of $\mathcal{O}(\NITEMS \DI^2)$. Moreover, using the dual representation of DPPs, sampling $\NBATCH$ elements among $\NITEMS$ has a time complexity of $\mathcal{O}(\NITEMS \DI \NBATCH^2 + \DI^2 \NBATCH^3)$~\citep{kulesza2012determinantal}, whereas the greedy approximation used for the maximization strategy (which is a NP-hard problem otherwise) has a naive time complexity of $\mathcal{O}(\NBATCH^2 \NITEMS)$. However, several works~\citep{gillenwater2012near,han2017faster,chen2018fast} improve upon this time complexity, by leveraging the fact that the logarithm of the determinant of the likelihood matrix in the DPP should be maximized 
and using linear algebra approximations. $\log \det$ is a submodular function but not monotone; instead of obtaining a $1 - 1/e$ approximation  of the exact solution~\citep{kempe2003maximizing}, proposed approaches yield 
an $1/4$ approximation~\citep{gillenwater2012near}, and run with a time complexity in $\mathcal{O}(\NITEMS \NBATCH^2)$~\citep{chen2018fast}. Moreover, \citet{mariet2019dppnet} introduce an approximation of a DPP--without the quality-diversity trade-off--by a deep learning model, named DPPNet, which could also be used to further empirically speed up the sampling and maximization of a DPP. Implementing our contributions with a DPPNet instead of a DPP might be a future venue for enhancing their use in real-life applications.

A recent paper by~\citet{kawashima2024family} introduces a family of distributions on sets of points which includes DPPs, called discrete kernel point processes (DKPPs), where the attraction and the repulsion of sampled points can be controlled explicitly. Given a continuous function $\psi : \bbR^{+} \mapsto \bbR$ and $\KERNEL_{\Omega,\Omega}$ a positive-semidefinite Hermitian matrix, the probability of sampling the subset $A \subseteq \Omega$ is given by $\mathbb{P}(\BATCH{}=A) \propto \exp\big(\TRACE(\psi(\KERNEL_{A,A})) \big)$. $\psi$ controls the parametrization between positive and negative correlations. DKPPs enable maximization, \IE finding the subset with highest probability, and sampling as well. However, using DPPs allows us to leverage existing implementations and algorithms. Furthermore, positive correlations between items are not necessarily linked to relevance, as the user information or context might also intervene in the quality score.

\subsection{Baselines for diversified recommendations}

This section gives an overview of the state-of-the-art on diversified recommender systems. We focused on Determinantal Point Processes (DPPs) in our paper to leverage the literature regarding fast implementation of sampling and maximization algorithms, and to incorporate some flexibility in the description of item similarity thanks to the use of kernel functions. Here, we do not consider prior works where the feedback model needs to be learned on the fly~\citep{radlinski2008learning,chao2015large,kathuria2016batched,nava2022diversified,hikmawati2024improve,park2025multi}.

As mentioned in Section~\ref{sec:introduction}, DPPs are popular in the field of recommender systems. Recently, \citet{ibrahim2025diversified} applied two DPPs (a traditional quality-diversity decomposition, and one with a linear kernel, without quality scores nor trade-off parameter) to output relevant and diverse recommendations of cultural goods to teenagers, again, in the sense of ``intrabatch'' diversity.  They reported an improved diversity (volume) in the recommendation with DPPs, to the price of a notable part of the relevance, both in offline and online/live experiments. \citet{wilhelm2018practical} (using a quality-diversity decomposition with a RBF kernel) also applied DPPs on video recommendation. Authors showed that DPPs, contrary to all other baselines, yielded an increase in the number of long user sessions, which are indicative of user satisfaction with the recommendations.

As claimed in the main text, the framework described in Section~\ref{sec:ALGO} allows us to recover many well-known DPPs from the literature. For instance, traditional quality-diversity decomposition~\citep{kulesza2010structured} can be obtained by setting $f^\text{QDDecomp}(\KERNEL, \BATCH{}, \HISTORY{\USER{}}{}) := \KERNEL_{\BATCH{},\BATCH{}}$ with $\lambda=0.5$, where $\KERNEL$ is the linear kernel $\KERNEL(\{\psi\}, \{\phi\}) = \psi^\intercal \phi$ or, equivalently, $\nu : x \mapsto x$. A conditional DPP~\citep{borodin2005eynard} can be described using $\lambda=0.5$ and $f^\text{CondDPP}(\KERNEL, \BATCH{},  \HISTORY{\USER{}}{}) := \KERNEL_{\BATCH{},\BATCH{}}-\KERNEL_{\BATCH{},\HISTORY{\USER{}}{}}(\KERNEL_{\HISTORY{\USER{}}{},\HISTORY{\USER{}}{}})^{-1}(\KERNEL_{\BATCH{},\HISTORY{\USER{}}{}})^\intercal$. Finally, an $\varepsilon$-greedy approach can be obtained by setting $f^\text{$\varepsilon$-greedy}(\KERNEL, \BATCH{}, \HISTORY{\USER{}}{}) := \text{I}_{\NBATCH}$ with $\lambda=0.5$ $\varepsilon$\% of the rounds (greedy phase), and setting $f^\text{$\varepsilon$-greedy}(\KERNEL, \BATCH{},  \HISTORY{\USER{}}{}) := \KERNEL_{\BATCH{},\BATCH{}}$ with $\lambda=0$ the other $(1-\varepsilon)$\% of rounds (exploratory phase). 

Recommender systems that build upon MMR, which is a strong baseline in the field of diverse recommendations, have also been investigated in the literature. Recently, \citet{li2024contextual} introduced Contextual Distillation Model (CDM), which trains a surrogate attention-based deep model of Maximal Marginal Relevance (MMR), for a diversity metric correlated to the scalar product of item and user vectors. The surrogate model estimates the probability of a given item of being ranked among the top-$\NBATCH$ items by the MMR score (see Section~\ref{sec:related_works}). Authors compared their contribution to a DPP~\citep{chen2018fast} and MMR, showing that the DPP had a better performance diversity-wise, but arguing that the latter have a quadratic time complexity in the number of items $\NITEMS$. However, we showed in our work that, provided some approximations (see Appendix~\ref{app:tractable}), we can actually run DPPs (\ALGO{} and the conditional DPP) and MMR on libraries of dozens of millions of items. Note that \citet{li2024contextual} only focused on what we call ``intrabatch diversity''.

In addition to~\citet{affandi2012markov}, there are other papers regarding the improvement of the diversity compared to a past sequence of recommended items. In addition to proposing a faster algorithm for the \textsc{MAXIMIZATION} strategy for DPPs, \citet{chen2018fast} also studies the case when the diversity is only required among items in a sliding time window. They use a standard quality-diversity decomposition for the likelihood matrix, with the \textsc{MAXIMIZATION} strategy, and modify the maximization problem such that only the last $\omega-1$ recommended items are taken into account in the greedy algorithm. However, their approach is suitable when a long series of interactions with the user occurs and when we allow the recommender system to forget the oldest recommendations. Yet, in practice, users interact only few times with the recommender system before settling on a recommendation or dropping out~\citep{ben2022modeling,gusak2025recommendation}. 

Finally, few papers on the literature deal with the issue of tuning the level of diversity in the recommendations to the user. Given a notion of usefulness (\IE relevance) and unexpectedness, which is somewhat related to diversity (see the previous paragraph). \citet{xu2024serendipitous} studies the problem of finding the good proportions of usefulness and unexpectedness in recommendations for each user after a step of retrieval of relevant items. The latter is closely related to the challenge of addressing the adaptive quality–diversity trade-off in Section~\ref{sec:adaptive_lambda}, albeit our procedure operates at any stage of recommendation (including retrieval) and does not feature a serendipity metric. The metric in~\citet{xu2024serendipitous} incorporates for each user a convex combination of long-term and short-term preferences as the usefulness-unexpectedness trade-off (``curiosity'') parameter. In that paper, this parameter is precomputed for each user, and cannot change anymore, unlike our approach.

\clearpage

\section{Practical implementation}\label{app:tractable} 

We resorted to the Python package DPPy~\citep{gautier2019dppy,calandriello2020sampling} for the implementation of the \textsc{SAMPLING} strategy, whereas we used the greedy algorithm in~\citet{chen2018fast} for the \textsc{MAXIMIZATION} approach. Please refer to the \texttt{requirements.txt} file in the code for the package versions. This section is motivated by the fact that, in Algorithm~\ref{alg:recom_setting}, the step of computation of the likelihood matrix for the DPP in Line 5 and of the recommendation set in Lines 6-7 might be expensive (at least in $\mathcal{O}(\NITEMS^3)$) when $\NITEMS$ is large when implemented naively. 

We also used sparse matrices as implemented in the Python package SciPy~\citep{virtanen2020scipy} with a denuding procedure: meaning that we could make the matrix sparser by rounding up values in matrices up to the $q^\text{th}$ decimal place to save memory. In practice, for our experiments, we did not have to round up values. However, for even larger libraries of items, this technique might prove useful. Below, we list the most computationally expensive parts of the code and explain the solutions considered to get a tractable implementation.

\subsection{Nystr\"{o}m approximations of the kernel function}\label{app:nystroem}

One important element of the code is that a dense matrix of size $\NITEMS \times \NITEMS$ should never be stored in full in the RAM. To comply with this rule, we extensively relied on the feature map $\nu$ associated with a kernel function $\KERNEL$ (which also facilitates the computations as if $\KERNEL$ were a linear function). To compute $\nu$, we use a Nystr\"{o}m approximation~\citep{nystrom1930praktische} that assumes that the rank of $\KERNEL_{\Omega,\Omega}$ is actually $C \ll \NITEMS$. Given a dimension $\DI' \ll \DI$, the Nystr\"{o}m approximation starts by subsampling a set of $C$ items $\mathcal{I}$ in $\Omega$ at random (without replacement), and then builds the eigendecomposition of the real-valued matrix $\KERNEL_{\Omega,\Omega}$, that is, the evaluation of the kernel function on the whole universe.
\begin{eqnarray*}
    \KERNEL_{\Omega,\Omega} & = & U \Lambda U^\intercal = \begin{bmatrix} U_{\mathcal{I}} \Lambda U_{\mathcal{I}}^\intercal & U_{\mathcal{I}} \Lambda U_{\neg \mathcal{I}}^\intercal\\
   U_{\neg \mathcal{I}} \Lambda U_{\mathcal{I}}^\intercal & U_{\neg \mathcal{I}} \Lambda U_{\neg \mathcal{I}}^\intercal \end{bmatrix} = \begin{bmatrix} \KERNEL_{{\mathcal{I}},{\mathcal{I}}} & \KERNEL_{{\mathcal{I}},{\neg \mathcal{I}}}\\
   \KERNEL_{{\mathcal{I}},\neg {\mathcal{I}}}^\intercal & \KERNEL_{{\neg \mathcal{I}},{\neg \mathcal{I}}} \end{bmatrix}\:,
\end{eqnarray*}

where $U_{\mathcal{I}}$ (respectively, $U_{\neg \mathcal{I}}$) is the part of the orthonormal basis of the eigendecomposition associated with items in $\mathcal{I}$ (resp., not in $\mathcal{I}$), and $\Lambda$ is the diagonal matrix of eigenvalues of $\KERNEL_{\Omega,\Omega}$.  When one needs to compute $\KERNEL_{\BATCH{},\BATCH{}}$ for any subset $\BATCH{} \subset \Omega$ of size $\NBATCH$, using the fact that $\KERNEL_{\BATCH{},\BATCH{}} = \big(\KERNEL_{\BATCH{},\mathcal{I}}^\intercal U_{\mathcal{I}}\Lambda^{-1/2}\big) \big(\KERNEL_{\BATCH{},\mathcal{I}}^\intercal U_{\mathcal{I}}\Lambda^{-1/2}\big)^\intercal$ by replacing $U_{\neg \mathcal{I}}$ by $\KERNEL_{\neg \mathcal{I}, \mathcal{I}} U_{\mathcal{I}} \Lambda^{-1}$. $\Lambda^{-1/2}$ consists in simply applying the square root function to the eigenvalues of $\KERNEL_{\Omega,\Omega}$, whereas $\KERNEL_{\BATCH{},\mathcal{I}} \in \bbR^{\NBATCH \times \DI'}$ where $\DI' \ll N$ can be explicitly evaluated. All in all, the corresponding $\nu$ function is $\nu(\BATCH{}) = \KERNEL_{\BATCH{},\mathcal{I}}^\intercal U_{\mathcal{I}}\Lambda^{-1/2}$ for any $\BATCH{} \subset \Omega$. 

As previously mentioned in the main text, the time complexity of the  Nystr\"{o}m approximation is $\mathcal{O}(\NITEMS (\DI')^2 + (\DI')^3)$. In practice, we use the implementation of the Nystr\"{o}m approximation in the Python package scikit-learn~\citep{pedregosa2011scikit}. The computation of $\mathcal{I}$, $U_{\mathcal{I}}$ and $\Lambda$ happens only once before any recommendation is made. We use $\DI'=100$ for the rank in the Nystr\"{o}m approximation. We note that, while we use random elements to build the approximation, other papers suggest to consider more representative points~\citep{tremblay2019determinantal}, for instance, those with maximum leverage score which would be more representative of the set. However, this approach can be costly (see Appendix~\ref{app:related_works}) which is why we elected to stick to random sampling.

Note that the computation of the volume (see metrics in Equation~\ref{eq:def_metrics}) also relies on the Nystr\"{o}m approximation, as the full matrix $\KERNEL_{\BATCH{},\BATCH{}}$ is not built. This accounts for the fact that, although Assumption~\ref{as:bounded_norm} mentions that the item embeddings have a $\ell_2$-norm equal to $1$, the actual computed volumes in the experiments in Section~\ref{sec:experiments} might exceed $1$. The Nystr\"{o}m approximation does not preserve the property on the norm of the item embeddings.

\subsection{Approximate matrix inversions and determinants}

For the conditional and the k-Markov~\citep{affandi2012markov} DPPs, there is a step of inversion of a positive definite matrix of potentially size $\NITEMS \times \NITEMS$, which is extremely costly in a naive implementation. As suggested by prior works~\citep{burian2003fixed}, we resort to Cholesky decompositions of the matrices that we wish to inverse. The Cholesky decomposition of a positive-definite matrix $M$ of size $n$ is $M = R^\intercal R$ where $R$ is an upper (or lower) triangular matrix of size $n$. Then, as $M^{-1} = R^{-1}(R^{-1})^\intercal$, it suffices to solve $RX = I_n$ in $X$ by forward substitution, because $R$ is a triangular matrix, to obtain the inverse of $M$. To compute the (log-)determinant of $M$ of size $n$, the following relation is true: $\log\det(M)=2\sum_{i=1}^n \log R_{i,i}$, because $R$ is a triangular matrix with positive diagonal elements~\citep{madar2015direct} and $\det R = \det R^\intercal$.

In our implementation, we use the Python package Scikit-Sparse.  For the conditional DPP, the Nystr\"{o}m approximation of the kernel function can be leveraged on top of the Cholesky decomposition (decomposing $\widetilde{A}^{1/2}\widetilde{A}^{1/2} \approx A$ instead of $A$, which is less expensive due to the low-rank assumption for $\widetilde{A}^{1/2}$), which is why the conditional DPP can be run on large data sets contrary to Markov DPP. Indeed, the theoretical time complexity of the Cholesky decomposition is in $\mathcal{O}(\NITEMS^3)$ as a general rule. However, for sparse matrices, a fill-reducing Cholesky decomposition can be computed by reordering rows to restrict the creation of new nonzero elements~\citep{duff2009combinatorial}, and for sparse matrices, the practical time complexity is typically much smaller than cubic in $\NITEMS$.

\subsection{Approximate matrix power}

Matrix power intervenes in Equation~\ref{eq:extension_qdd} which describes the likelihood matrix for a DPP of the DQD family for $\lambda \neq 1/2$, where the power is possibly real-valued. The exact computation would require computing the whole set of eigenvalues and eigenvectors, which is quickly intractable as $\NITEMS$ increases. Similarly to Appendix~\ref{app:nystroem}, we make a low-rank assumption and consider the truncated Singular Value Decomposition (SVD) of rank $r \ll \NITEMS$ of a real-valued matrix $M$ where $M$ is a square real matrix of size $\NITEMS$. The SVD yields $M=U \Lambda V^\intercal$, where $U, V \in \bbR^{r \times \NITEMS}$ are orthogonal matrices. Then we compute $V \Lambda^p V^\intercal$ as a proxy for $M^p$, that is, $M$ multiplied $p$ times where $p \in \bbR^+$. Indeed, 
$M^p \approx (V^\intercal \Lambda V)^p=\underbrace{(V^\intercal \Lambda V)(V^\intercal \Lambda V)\dots(V^\intercal \Lambda V)}_{\text{$p$ times}}=V^\intercal \Lambda^p V,$ since $VV^\intercal=I_r$ and $V^\intercal V=I_\NITEMS$. This operation has a time complexity in $\mathcal{O}(\NITEMS r^2)$ instead of the naive time complexity $\mathcal{O}(\NITEMS^3)$ which is dominated by the computation of the full eigendecomposition. Moreover, if we need to find $L$ instead, where $LL^\intercal=M^p$, we can output $L=V \sqrt{\Lambda^p}$ where $\sqrt{\Lambda^p}$ is such that $\sqrt{\Lambda^p}\sqrt{\Lambda^p}=\Lambda^p$. Note that, for the diagonal matrix $\Lambda^p$, $\sqrt{\Lambda^p}$ is the result of the square root function applied element-wise to the diagonal elements of $\Lambda^p$. We choose $r=\NITEMS-1$ if $\NITEMS\leq 1\,000$ and $r=100$ otherwise.

\subsection{Approximate closest-neighbor finding algorithms}

\ALGO{} requires the computation of the maximum cosine distance between feature maps of items from the current batch and the user history (see Equation~\ref{eq:ALGO}). Obviously, computing all cosine distances between any item and all items in the user history would be expensive since it is performed for all items in $\Omega$. As mentioned in Section~\ref{sec:ALGO}, many approaches allow us to retrieve the closest neighbor of a point in a set. Our implementation uses FAISS~\citep{douze2024faiss}, which is suitable for large data sets. In practice, we do not even need to build a new tree for each user and each recommendation round, but only a single tree on all items in $\Omega$. Indeed, queries to FAISS can be made while ignoring a subset of items when querying for the closest neighbor (in our case, ignoring all items which are not in the current user's history).

\subsection{Retrieving embeddings}

Finally, our implementation cannot afford to store the full item embedding matrix $\Phi=[\ITEM{i}, i \in \Omega] \in \bbR^{\NITEMS \times \DI}$ in the RAM. First, we store the user histories in files, and update them in-place. Second, we store, access and read the item embedding matrix by batches in memory with .csv files. Moreover, for the largest data set of our paper (the synthetic data set SYNTHETIC15M with five millions items), we use the .npy binary format instead of .csv for faster input/output operations.  

\clearpage

\section{Adaptive quality-diversity trade-off}\label{app:adaptive}

We recall here the online learning game for the adaptive quality-diversity trade-off (see Section~\ref{sec:adaptive_lambda}). The player is the recommender system, and Nature is the interaction with \textit{one} user $(\USER{}, \HISTORY{\USER{}}{t})$ querying the recommender system at time $t$~\citep{auer2002nonstochastic,de2014follow}. At round $t \leq \HORIZON$ of the game, the recommender system chooses a value of $\lambda^t \in [0,1]$. Then, the recommender system makes a recommendation $\BATCH{t}$ to user $\USER{}$ and receives a loss value related to the quality-diversity trade-off achieved, depending on the vector of feedback $\bm{y}^t$ from Nature. The recommender system should use this loss to update and use $\lambda^{t+1} \in [0,1]$ at the next round. The game ends after $\HORIZON$ interactions with the same user.

\subsection{Implementation of the adaptive diversity-tuning procedure}\label{subapp:adaptive_implementation} 

Applying the AdaHedge algorithm~\citep{de2014follow} to our problem, we consider two experts. One favors diversity, whereas the other prefers quality. We denote the posterior weight vector associated with those two experts $\bm{\lambda}^t := [1-\lambda^t, \lambda^t]$ at time $t$. In our case, we consider a gain instead of a loss function--since the optimization problem in Equation~\ref{eq:regret_adaptive} is a maximization problem. The function (to maximize) at time $t$ with $\lambda^t$ is 
\begin{eqnarray}\label{eq:extension_qdd_lambdat}
    \SETSCORE{t}^{\lambda^t}_{f,\bm{y}^t}(\BATCH{t};\USER{t}) & =& 4(1-\lambda^t) \log \VOLUME{f(\KERNEL, \BATCH{t}, \HISTORY{\USER{t}}{t})} + 4\lambda^t \log \det\big(\text{diag}(\{y^t_k\}_{k \leq \NBATCH})\big) \\
    & =& 4(1-\lambda^t) \log \VOLUME{f(\KERNEL, \BATCH{t}, \HISTORY{\USER{t}}{t})} + 4\lambda^t \sum_{k \leq \NBATCH} \log y^t_k \:, \nonumber
\end{eqnarray}
where $f=f^\text{BDR}$ (note however that this approach straightforwardly works for any DPP of the DQD family as described in Equation~\ref{eq:extension_qdd}). Now, since $\SETSCORE{t}^{\lambda^t}_{f,\bm{y}^t}(\BATCH{t};\USER{t})$ is clearly linear in $\lambda^t$ (hence concave), we can apply the gradient trick as follows
\begin{eqnarray*}
   \mathcal{R}(\HORIZON; \USER{}) & := & \max_{\lambda \in [0,1]} \sum_{t \leq \HORIZON} \SETSCORE{t}^\lambda_{f,\bm{y}^t}(\BATCH{t};\USER{}) - \SETSCORE{t}^{\lambda^t}_{f,\bm{y}^t}(\BATCH{t};\USER{})\\
   & \leq & \max_{\lambda \in [0,1]} \sum_{t \leq \HORIZON} (\lambda-\lambda^t) \nabla_\lambda \SETSCORE{t}^{\lambda^t}_{f,\bm{y}^t}(\BATCH{t};\USER{})\:.
\end{eqnarray*}
The gradient $\nabla_\lambda \SETSCORE{t}^{\lambda}_{f,\bm{y}^t}(\BATCH{t};\USER{})$ with respect to $\lambda \in [0,1]$ is straightforward to compute
\begin{eqnarray}\label{eq:gradient_qdd_lambdat}
    \nabla_{\lambda} \SETSCORE{t}^{\lambda}_{f,\bm{y}^t}(\BATCH{t};\USER{t})  & := & - 4 \log \VOLUME{f(\KERNEL, \BATCH{t}, \HISTORY{\USER{t}}{t})} + 4 \sum_{k \leq \NBATCH} \log y^t_k \:.
\end{eqnarray}
Then, we obtained an upper bound on the quantity that we want to optimize in the form of a inner product. This is the motivation for using as a gain function at time $t$
\[g_t : \bm{\lambda} \in \triangle_2 \mapsto [-\nabla_{\lambda} \SETSCORE{t}^{\lambda}_{f,\bm{y}^t}(\BATCH{t};\USER{t}), \nabla_{\lambda} \SETSCORE{t}^{\lambda}_{f,\bm{y}^t}(\BATCH{t};\USER{t})]\:,\]
 where $\triangle_2$ is the simplex of dimension $2$. Note that the first coordinate of $g_t$ $g_t(\bm{\lambda})_1$ is actually the gradient of $\SETSCORE{t}^{\lambda}_{f,\bm{y}^t}(\BATCH{t};\USER{})$ with respect to $1-\lambda$ (proven with a change of variable). Then in the AdaHedge learner, at time $t$, we update the posterior weight with 
\[-\langle \bm{\lambda}^t, g_t(\bm{\lambda}) \rangle = -\big((1-\lambda^t)(-g_t(\bm{\lambda})_2)+\lambda^tg_t(\bm{\lambda})_2\big) = (1-2\lambda^t)g_t(\bm{\lambda})_2 = (1-2\lambda^t)\nabla_{\lambda} \SETSCORE{t}^{\lambda}_{f,\bm{y}^t}(\BATCH{t};\USER{t})\:,\]
which corresponds to Line 8 in Algorithm~\ref{alg:recom_setting}. Then, we obtain a new value $\bm{\lambda}^{t+1} \in \triangle_2$, from which we extract the second coordinate $\lambda^{t+1}$ for the next round.

\subsection{Upper bound on the regret}\label{subapp:adaptive_bounds} 

\begin{theorem}{\textnormal{Upper bound on the regret incurred by the adaptive diversity tuning procedure (Theorem~\ref{th:regret_UB} in the main text).}} An upper bound on the regret $\mathcal{R}(\HORIZON)$ (Equation~\ref{eq:regret_adaptive}) incurred by the adaptive strategy for tuning the level of diversity $\lambda \in [0,1]$ for user $\USER{}$ over $\HORIZON$ rounds of recommendations is 
\[\mathcal{R}(\HORIZON) \leq 2\delta_T\sqrt{T\log(2)}+16\delta_T(2+\log(2)/3)\:,\]
where $\delta_T := 8 \max_{t \leq T} \log \frac{M_t^B}{a_t}$, $a_t := \VOLUME{f(k,S^t,H^t)}$ and $M_t := \max_{i \leq B} y^t_i$.
\end{theorem}

\begin{proof}
We defined the regret $\mathcal{R}(\HORIZON; \USER{})$ for user $\USER{}$ at time $\HORIZON$ in Equation~\eqref{eq:regret_adaptive}, with the gain function $G_t : \lambda \in [0,1] \mapsto \SETSCORE{t}^\lambda_{f^\text{BDR},\bm{y}^t}(\BATCH{t};\USER{})$ for any $t \leq \HORIZON$. $G_t$ 
is linear (and then concave) in $\lambda$ for all $t \leq \HORIZON$. Moreover, remember that we applied AdaHedge with the gain function $g_t$ at time $t$ in Line 8 in Algorithm~\ref{alg:recom_setting}, where 
\[\forall \lambda \in [0,1] \ , \quad g_t(\lambda) := (1-2\lambda)\underbrace{\left( -4\log(\VOLUME{f(k,S^t,H^t)})+4\sum_{i \leq \NBATCH} \log y^t_i\right)}_{= C_t}\;.\]

by using the computation from Subsection~\ref{subapp:adaptive_implementation}. We denote $\bm{\lambda} := [1-\lambda, \lambda]^\intercal \in \triangle_2$, where $\triangle_2$ is defined as the simplex of dimension 2 $\{ \bm{p} \in [0,1]^2 \mid p_1 + p_2 = 1\}$. Then we can rewrite the expression of $g_t$ as 
\[\forall \lambda \in [0,1] \ , \quad g_t(\lambda) = -\langle \bm{\lambda}, [C_t, -C_t]^\intercal \rangle\;.\]

Then, combining the concavity of $G_t$ in $\lambda$, the fact that $C_t$ is constant in $\lambda$, and the definition of $g_t$,
\begin{eqnarray*}
    \forall \lambda \in [0,1] \ , \sum_{t \leq \HORIZON} G_t(\lambda) - G_t(\lambda^t) & \leq & \sum_{t \leq \HORIZON} \langle \bm{\lambda}-\bm{\lambda}^t, [C_t, -C_t]^\intercal \rangle\\
    \implies \mathcal{R}(\HORIZON; \USER{}) & \leq &\max_{\lambda \in [0,1]} \sum_{t \leq \HORIZON} g_t(\lambda) - g_t(\lambda^t)\;.
\end{eqnarray*}

Then, we apply Theorem $8$ and Corollary $17$ from~\citet{de2014follow} to the loss function $-g_t$, using the fact that $\lambda \in [0,1]$,

\[R^\text{adapt}(T;\textbf{h}) \leq  2\delta_{\HORIZON}\sqrt{\HORIZON\log(2)}+16\delta_{\HORIZON}(2+\log(2)/3)\;, \text{ where } \delta_{\HORIZON} := 2\max_{t \leq \HORIZON} \max(C_t , -C_t)\;.\]

What remains is to evaluate $\delta_\HORIZON$. According to Assumptions~\ref{as:kernel}-\ref{as:bounded_norm}, there are two positive constants $m_t$ and $M_t \geq 1$ such that $0 < m_t \leq y^t_i \leq M_t$, and $0 < a_t := \text{vol}(f(k,S^t,H^t)) \leq 1$. Using the computations in Subsection~\ref{subapp:adaptive_implementation}, the following inequalities hold
\begin{eqnarray*}
    \forall t \leq \HORIZON \ ,  \quad -4\log a_t+4\NBATCH\log m_t &\leq& C_t \leq -4\log a_t+4\NBATCH\log M_t\\
    4\log (m_t^\NBATCH/a_t) &\leq & C_t \leq 4\log (M_t^\NBATCH/a_t)\;.
\end{eqnarray*}
This means that $\delta_{\HORIZON} \leq 8\max_{t \leq \HORIZON} \log (\max(m_t, M_t)^\NBATCH/a_t) = 8\max_{t \leq \HORIZON} \log \frac{M_t^\NBATCH}{a_t}$.
\end{proof}

Note that Assumption~\ref{as:feedback_model} can be replaced by an assumption on the positiveness of \textit{observed} feedback values, hence this result still hold true when we do not have access to a feedback model.
 





\clearpage

\section{Supplementary material on the experimental study}

All experiments have been run on a remote server (configuration:
processor Intel Xeon Processor (Skylake, IBRS), 12 cores @2.4GHz, RAM 23GB). 

\subsection{Generation of synthetic data sets}\label{app:synthetic}

The synthetic data sets aim at providing an easy approach to test the scalability of the algorithms and the diversity in recommendations. The key idea is that $\DI$-dimensional item embeddings ($\ITEM{1},\dots,\ITEM{\NITEMS}$) and user contexts ($\USER{1},\dots,\USER{N_u}$) are sampled at random (hence $\DU=\DI$), and the feedback value for the item-user pair $(\ITEM{}, \USER{})$ is $((\ITEM{})^\intercal \USER{}+1)/2 \in [0,1]$. We actually generate $\NBATCH$ groups of items, with the expectation that, when outputting a batch of $\NBATCH$ recommendations, a good recommender system for diversity would sample elements from each group.

Given a desired number of items $\NITEMS$ and a batch size $\NBATCH$, we first sample at random and independently $\lceil \NITEMS/\NBATCH \rceil \times \DI$ values from $\mathcal{N}(0,2)$ to build a first matrix of $\ell_2$-normalized item embeddings $\Phi^1 := [\ITEM{1},\ITEM{2},\dots,\ITEM{\lceil \NITEMS/\NBATCH \rceil}] \in \bbR^{\lceil \NITEMS/\NBATCH \rceil \times \DI}$. Then, for the $\ell^{\text{th}}$ group, $\ell=2,\dots,\NBATCH$, we construct the matrix $\Phi^\ell := \Phi^i + 0.01\ell \bm{1}_{\lceil \NITEMS/\NBATCH \rceil \times \DI}$, where $\bm{1}_{n \times d}$ is the matrix of size $n \times d$ which coefficients are ones. Finally, we concatenate all (renormalized) matrices $\Phi^1,\Phi^2,\dots,\Phi^\NBATCH$ to obtain $\Phi \in \bbR^{\NITEMS \times \DI}$. For users, given a maximum number of users $N_u$, $N_u \times \DI$ values are sampled from $\mathcal{N}(0,1)$ to build the user contexts, after $\ell_2$-normalization.

\subsection{Training on real-life data sets}\label{app:real_life}

In this section, we delve more into detail concerning the two real-life data sets for movie recommendation and drug repurposing. In both cases, the important part is to learn the feedback model $\FREWARD$ based on prior user-item interactions (\EG previously rated movies or terminated clinical trials, collected offline by browsing clinical trial registries). The feedback model is then used in the computation of the likelihood matrix in Line 5 in Algorithm~\ref{alg:recom_setting}. See the \texttt{README.md} file for instructions to download the files related to those data sets.

\paragraph{Movie recommendation: MovieLens.} The MovieLens data set~\citep{harper2015movielens} that we considered has 9,725 items, 610 users and 100,837 nonzero ratings in $\{1,2,\dots,5\}$. The higher the rating, the more the user liked the movie. Metadata about the items/movies are the movie title, year of release, and the genre tags (\EG comedy, action). First, $512$-dimensional item embeddings of norm $1$ are built by applying Universal Sentence Encoder~\citep{cer2018universal} to the movie metadata, using Python package Tensorflow Hub~\citep{abadi2016tensorflow}. Second, we split the list of item-user pairs with nonzero ratings into a training set ($75\%$) and a testing set ($25\%$) at random.  We fit a SVD model on the training set using the Python package Mangaki~\citep{vie2015mangaki}, which fills the zeroes in the ratings matrix of MovieLens. We computed the Root Mean Squared Error on the testing set, yielding RMSE$=0.96$. Knowing that the ratings are in $\{1,2,\dots,5\}$, it means that, in average, the mistake made by the SVD model on the predicted rating is not enough to accidentally classify a very bad movie as a good one. We considered a classification threshold at $\tau=2.5$.

\paragraph{Trust network: Epinions.} The Epinions data set~\citep{leskovec2010signed} contains users' ratings (from 1 to 5) on items. We use the same approach as for MovieLens to define the feedback model. To make it run smoothly across all baselines, we only kept users with at least 50 ratings, and items with at least 100 ratings, and used as item embedding matrix the first factor of the SVD, for a total of 160,417 items.

\paragraph{Drug repurposing: Cdataset~\citep{luo2016drug}, DNdataset~\citep{martinez2015drugnet}, Fdataset~\citep{gottlieb2011predict}, Gottlieb~\citep{gao2022dda}, LRSSL~\citep{liang2017lrssl}, and PREDICT~\citep{reda2023predict}.} Those data sets comprise information about drugs and diseases, along with the status of Phase 2 or 3 clinical trials involving each drug-disease pair: 0 means that the pair has not been tested in a Phase 2 or 3 clinical trial, 1 means that the clinical trial was successful in showing that the drug has a therapeutic effect on the disease, and -1 means that the clinical trial failed (\EG low accrual, emergence of adverse side effects). There are six types of drug features and five types of disease features across all data sets, corresponding to the similarity of a drug (respectively, a disease) to another drug (resp., disease) in the data set with respect to each type of feature. We give the description and relevant metrics regarding each drug repurposing data set in Table~\ref{tab:datasets}. 
We fitted the Heterogeneous Attention Network (HAN) algorithm~\citep{wang2019heterogeneous} as feedback model to those data sets, using the Python package stanscofi~\citep{reda2024stanscofi}, splitting nonzero ratings at random into training ($80\%$) and testing ($20\%$) sets. The Area Under the Curve on the testing set for the PREDICT data set (the richest data set in terms of features) was $0.92$.

 \begin{table}[t]
  \caption{Datasets in the benchmark. They correspond to the number of drugs and diseases involved in at least one nonzero drug-disease association. The rating matrix in the \texttt{Fdataset} comes from~\citet{gottlieb2011predict}, whereas the drug and disease features are from~\citet{luo2016drug}.}
  \label{tab:datasets}
  \begin{center}
    \begin{small}
      \begin{sc}
        \begin{tabular}{lcccccc}
          \toprule
         Name  & $\NITEMS$         & $\DI$      & Nb. users & $\DU$ & Nb. +1's & Nb. -1's  \\
          \midrule
     Cdataset & 663 & 663 & 409 &  409 	& 2,532 	& 0\\
    DNdataset &  550 &  1,490 & 360 & 4,516 	& 1,008 	& 0 	\\
     Fdataset & 593&  593 & 	313 &  313 & 	1,933 	& 0 	 \\
     Gottlieb &  593 &  1,779 &  313 & 313	& 1,933 	& 0 	\\
    LRSSL & 763 & 2,049 	& 681 & 681 	& 3,051 	& 0 	\\
      PREDICT &  1,014 &  1,642  & 941 & 1,490 	& 4,627 	& 132 \\
          \bottomrule
        \end{tabular}
      \end{sc}
    \end{small}
  \end{center}
  \vskip -0.1in
\end{table}
 
\subsection{Benchmark on xQuAD and Deep DPP}\label{app:benchmark}

xQuAD~\citep{santos2010exploiting} relies on the definition of subqueries in the ranking score for items. This ranking score is similar to the one for MMR~\citep{carbonell1998use}, except for not usng any diversity metrics, but considering relevance scores for items fitting each subquery. This means in particular that xQuAD only explicitly relies on the feedback model $\FREWARD{}{}$. To make the comparison fairer with other baselines relying on the diversity metrics, we consider as sub-query generation procedure the selection of the items in the history with a cosine similarity higher than $\alpha$ (similarly to what was defined for our algorithm \ALGO). To implement this, we resort to FAISS trees as described in Equation~\ref{eq:ALGO}. Then, the main difference between xQuAD and BDivRec is that xQuAD only considers the relevance scores of items and similar items in the history in its ranking score. 

Deep DPP~\citep{gartrell2018deep} learns a low-rank factor $V$ of the likelihood matrix for a DPP, meaning that the final likelihood matrix is $L=VV^\intercal$. This learning is based on sets of observed subsets of items, and backpropagation of the likelihood function. In our benchmark, we define the observed subsets needed for learning the likelihood matrix as batches of items of at most $\NBATCH$ in the user history. Contrary to \ALGO, Deep DPP does not include the estimated feedback values nor the diversity values in the ranking score. Moreover, if the user history is empty, we set the likelihood matrix to the identity matrix.

We perform tests on a synthetic data set with 1,500 items (SYNTHETIC1500), the Epinions data set~\citep{leskovec2010signed}, and on all drug repurposing data sets in Appendix~\ref{app:real_life}. Those results are shown in Table~\ref{tab:benchmark_app}. The results are somewhat expected, as xQuAD do not take into account explicitly the similarity scores between items and the corresponding algorithm has almost the same structure as MMR. Deep DPP leverages the power of DPPs to perform well (local) diversity-wise, but is worst at relevance. BDivRec clearly improves upon all baselines either in terms of diversity (xQuAD) or relevance (Deep DPP), highlighting the quality-diversity tradeoff we aimed for. However, note that xQuAD and Deep DPP were not developed for the optimization of the quality-global diversity tradeoff.

\begin{table}[tb!]
  \caption{Benchmark on real-life data sets: Cdataset (4 users, $\NBATCH=3$), DNdataset (4 users, $\NBATCH=3$), Epinions (3 users, $\NBATCH=3$: DeepDPP could not be run on Epinions, as the number of items is too large to compute the likelihood function in memory), Fdataset (4 users, $\NBATCH=3$), LRSSL (4 users, $\NBATCH=3$), MovieLens (4 users, $\NBATCH=3$), 
  PREDICT (4 users, $\NBATCH=3$) and SYNTHETIC1500 (4 users, $\NBATCH=5$).}
  \label{tab:benchmark_app}
  \begin{center}
    \begin{small}
      \begin{sc}
        \begin{tabular}{lrrrrrr}
          \toprule
           Cdataset         & \text{rel}  $\uparrow$  & \text{prec}  $\uparrow$ & \text{div}$^L$  $\uparrow$  & \text{div}$^G$ $\uparrow$ &  \text{div}$^+$ $\uparrow$ &  time $\downarrow$ \\
          \midrule
        Deep DPP& 0.36$\pm$ 0.03 & 0.35$\pm$ 0.05&\underline{0.48$\pm$ 0.01} & \underline{0.16$\pm$ 0.02} &  0.32$\pm$ 0.04 & 0.37$\pm$ 0.02\\
        MMR &     \textbf{0.91$\pm$ 0.01} & \textbf{1.00$\pm$ 0.0} &0.36$\pm$ 0.0 &0.14$\pm$ 0.01 &0.36$\pm$ 0.0 & \underline{0.35$\pm$ 0.01}\\
        xQuAD & \textbf{0.91$\pm$ 0.01}  & \textbf{1.00$\pm$ 0.0}  &  \underline{0.37$\pm$ 0.01} &0.14$\pm$ 0.01 & 0.37$\pm$ 0.01&0.54$\pm$ 0.02\\
        \textbf{\ALGO}& \underline{0.79$\pm$ 0.01} &  \underline{0.98$\pm$ 0.0} & \textbf{0.50$\pm$ 0.02} & \textbf{0.22$\pm$ 0.02} &\textbf{0.52$\pm$ 0.02}&   \textbf{0.25$\pm$ 0.0}\\
             \midrule
             DNdataset         & & &  & &  & \\
             \midrule
        Deep DPP& 0.23$\pm$ 0.06 & \underline{0.24$\pm$ 0.06} &0.81$\pm$ 0.01 & 0.72$\pm$ 0.01 &  \underline{0.23$\pm$ 0.06} & 1.58$\pm$ 0.08\\
        MMR &  \textbf{0.33$\pm$ 0.06}  &\textbf{0.25$\pm$ 0.07}& \underline{0.89$\pm$ 0.03} &\underline{0.86$\pm$ 0.03} &0.19$\pm$ 0.05 & \underline{1.57$\pm$ 0.07}\\
        xQuAD & \textbf{0.33$\pm$ 0.06} & \textbf{0.25$\pm$ 0.07}   &   0.82$\pm$ 0.01 &0.77$\pm$ 0.01&0.22$\pm$ 0.06&1.61$\pm$ 0.07 \\
        \textbf{\ALGO}& \underline{0.31$\pm$ 0.06} &  \textbf{0.25$\pm$ 0.07} & \textbf{1.01$\pm$ 0.0} & \textbf{0.98$\pm$ 0.01} &\textbf{0.25$\pm$ 0.07} &   \textbf{1.37$\pm$ 0.04}\\
             \midrule
             Epinions         && &   & &  & \\
             \midrule
            DeepDPP & -- & -- & -- & -- & -- &  --\\
            MMR &     \textbf{0.04$\pm$ 0.01} & 0.0$\pm$ 0.0  & 0.0$\pm$ 0.0  &0.0$\pm$ 0.0  &0.0$\pm$ 0.0 & \textbf{3.01$\pm$ 0.15}\\
            xQuAD & \textbf{0.04$\pm$ 0.01}  &0.0$\pm$ 0.0   &  0.0$\pm$ 0.0  &0.0$\pm$ 0.0 & 0.0$\pm$ 0.0 & 51.94$\pm$ 3.11\\
            \textbf{\ALGO}& \underline{0.02$\pm$ 0.01}&  0.0$\pm$ 0.0 & 0.0$\pm$ 0.0 & 0.0$\pm$ 0.0 & 0.0$\pm$ 0.0 &  \underline{17.15$\pm$ 0.56}\\
             \midrule
             Fdataset         & & &  & &  & \\
             \midrule
                Deep DPP& 0.48$\pm$ 0.03 & 0.55$\pm$ 0.05 &\underline{0.47$\pm$ 0.01} & 0.12$\pm$ 0.01  & \underline{0.59$\pm$ 0.06} &  0.28$\pm$ 0.01\\
        MMR &  \textbf{0.95$\pm$ 0.01 }  &\textbf{ 1.0$\pm$ 0.0 } & 0.35$\pm$ 0.01 &\underline{0.16$\pm$ 0.01 }&0.35$\pm$ 0.01 & \underline{0.27$\pm$ 0.01} \\
        xQuAD & \underline{0.93$\pm$ 0.01} &  \textbf{ 1.0$\pm$ 0.0 }    &   0.37$\pm$ 0.01 &0.14$\pm$ 0.01&0.37$\pm$ 0.01&0.44$\pm$ 0.01 \\
        \textbf{\ALGO}&0.81$\pm$ 0.01 & \textbf{ 1.0$\pm$ 0.0 }& \textbf{0.64$\pm$ 0.05} & \textbf{0.24$\pm$ 0.02} &\textbf{0.64$\pm$ 0.05} &  \textbf{0.19$\pm$ 0.01}\\
             \midrule
             Gottlieb         & & &  & &  & \\
             \midrule
                Deep DPP& 0.50$\pm$ 0.03 & 0.58$\pm$ 0.06 &0.64$\pm$ 0.01 & 0.19$\pm$ 0.02 & 0.55$\pm$ 0.05 &  \textbf{0.44$\pm$ 0.01}\\
        MMR & \textbf{0.94$\pm$ 0.01} &\textbf{1.00$\pm$ 0.0}& 0.64$\pm$ 0.02 &\underline{0.37$\pm$ 0.03}&0.64$\pm$ 0.02&  \underline{0.72$\pm$ 0.02} \\
        xQuAD & \underline{0.93$\pm$ 0.01}& \underline{0.99$\pm$ 0.0}    &   \underline{0.65$\pm$ 0.01} &0.36$\pm$ 0.03&\underline{0.65$\pm$ 0.01} &0.88$\pm$ 0.02 \\
        \textbf{\ALGO}&0.82$\pm$ 0.01 & 0.97$\pm$ 0.02& \textbf{0.76$\pm$ 0.01} & \textbf{0.40$\pm$ 0.02} & \textbf{0.78$\pm$ 0.01 }& \textbf{0.44$\pm$ 0.01}\\
             \midrule
             LRSSL         &  & & & & &  \\
             \midrule
        Deep DPP& 0.39$\pm$ 0.02 & \underline{0.40$\pm$ 0.03} &\textbf{0.90$\pm$ 0.01} & 0.53$\pm$ 0.05 & 0.64$\pm$ 0.03 &  \underline{0.48$\pm$ 0.02}\\
        MMR & \textbf{0.97$\pm$ 0.0} & \textbf{1.00$\pm$ 0.0} &0.82$\pm$ 0.02 &\underline{0.55$\pm$ 0.05} &\underline{0.82$\pm$ 0.02}&  0.53$\pm$ 0.02 \\
        xQuAD &\textbf{0.97$\pm$ 0.0} & \textbf{1.00$\pm$ 0.0}   &  0.80$\pm$ 0.02 &0.53$\pm$ 0.04&0.80$\pm$ 0.02 &0.71$\pm$ 0.02\\
        \textbf{\ALGO}& \underline{0.83$\pm$ 0.01} & \textbf{1.00$\pm$ 0.0} & \underline{0.86$\pm$ 0.01} & \textbf{0.56$\pm$ 0.04} & \textbf{0.86$\pm$ 0.01} & \textbf{0.36$\pm$ 0.01}\\
             \midrule
             MovieLens      & &    &  & &  & \\ 
             \midrule
            Deep DPP & 1.25$\pm$ 0.03 &0.09 $\pm$0.01 & \textbf{0.85 $\pm$0.00} & 0.02$\pm$ 0.00 &  0.22$\pm$ 0.02& 502.63 $\pm$23.86\\
            MMR&     \underline{  3.78$\pm$ 0.08} & \underline{0.84 $\pm$0.03} &  0.74 $\pm$0.02 &  \textbf{0.06$\pm$ 0.01} & \textbf{0.77$\pm$ 0.03}& \underline{27.25 $\pm$0.35}\\
            xQuAD &  \textbf{ 3.85$\pm$ 0.06} & \textbf{0.94 $\pm$0.01} & 0.65 $\pm$0.02 & \underline{0.05 $\pm$0.00} & \underline{0.68 $\pm$0.02}& 42.66 $\pm$0.33\\
            \textbf{\ALGO}& 3.49$\pm$ 0.09  & 0.76 $\pm$0.04 & \underline{0.80 $\pm$0.01 } & \textbf{0.06$\pm$ 0.01} &  0.60$\pm$ 0.03 & \textbf{26.83 $\pm$0.34}\\
             \midrule
            PREDICT        & & &  & &  & \\
             \midrule
            Deep DPP& 0.21$\pm$ 0.02&0.17$\pm$ 0.03&0.64$\pm$ 0.03&0.35$\pm$ 0.03&0.30$\pm$ 0.04&\textbf{0.73$\pm$ 0.03} \\ 
MMR&0.66 $\pm$0.02&0.86 $\pm$0.01&\underline{0.69 $\pm$0.02}&\textbf{0.60 $\pm$0.02}&\underline{0.70$\pm$ 0.02}&1.88$\pm$ 0.01 \\ 
            xQuAD & \textbf{0.79$\pm$ 0.02} & \underline{0.92$\pm$ 0.02} & 0.66$\pm$ 0.03 & 0.47$\pm$ 0.03 & 0.64$\pm$ 0.02 & \underline{1.00$\pm$ 0.06}\\ 
            \textbf{\ALGO}&\underline{0.76$\pm$ 0.02}&\textbf{0.93$\pm$ 0.01}&\textbf{1.02$\pm$ 0.02}&\underline{0.57$\pm$ 0.02}&\textbf{1.02$\pm$ 0.02}&1.15$\pm$ 0.01 \\
          \midrule
            SYNTHETIC1500       & &  &  & &  & \\
             \midrule
           Deep DPP& 0.50$\pm$ 0.0&0.47$\pm$ 0.01&\textbf{1.00$\pm$ 0.0}&\underline{0.32$\pm$ 0.0}&\underline{0.98$\pm$ 0.0}&0.70$\pm$ 0.01 \\ 
           MMR & \textbf{0.65$\pm$ 0.0} & \textbf{1.00$\pm$ 0.0} & \underline{0.00$\pm$ 0.0} & 0.00$\pm$ 0.0 & 0.00$\pm$ 0.0 & \textbf{0.03$\pm$ 0.00}\\ 
           xQuAD & \textbf{0.65$\pm$ 0.0} & \textbf{1.00$\pm$ 0.0} & \underline{0.00$\pm$ 0.0} & 0.00$\pm$ 0.0 & 0.00$\pm$ 0.0 & 0.76$\pm$ 0.01\\ 
          \textbf{\ALGO}&\underline{0.53$\pm$ 0.0}&\underline{0.82$\pm$ 0.02}&\textbf{1.00$\pm$ 0.0}&\textbf{0.94$\pm$ 0.00}&\textbf{1.01$\pm$ 0.00}&\underline{0.06$\pm$ 0.00}\\
          \bottomrule
        \end{tabular}
      \end{sc}
    \end{small}
  \end{center}
  \vskip -0.1in
\end{table}

\clearpage

\section{Complementary experiments}\label{app:compl_experiments} 

We report and describe here the comparison between the \textsc{SAMPLING} and \textsc{MAXIMIZATION} strategies, and the problem of scalability of Markov DPPs~\citep{affandi2012markov}. All results are shown in Tables~\ref{tab:compl_experiments}-\ref{tab:compl_experiments2}.

We apply both the \textsc{SAMPLING} and \textsc{MAXIMIZATION} strategies for all recommender systems on the SYNTHETIC750 data set. These results confirm that, for the quality-diversity trade-off, \ALGO{} with the \textsc{MAXIMIZATION} strategy and MMR are the top contenders. Moreover, they show that the \textsc{SAMPLING} indeed encourages (global) diversity as evidenced by the $\text{div}^\text{G}$ values and as reported in prior work~\citep{kathuria2016batched}. However, it results in a large loss in precision ($\text{prec}$) and relevance ($\text{rel}$), it is more time-consuming, and does not quite achieve the best quality-diversity trade-off ($\text{div}^{+}$). The \textsc{SAMPLING} strategy might be useful in the case when the recommender system must not recommend the same batch for users with the same embeddings and history, and when we are willing to recommend possibly irrelevant items. However, when the focus is on the optimization of the quality-diversity trade-off and computational efficiency, the \textsc{MAXIMIZATION} strategy might be more suitable.

We also ran the same type of experiments 
on SYNTHETIC30k, SYNTHETIC3M and SYNTHETIC15M. Beyond $5\,000$ items, MarkovDPP~\citep{affandi2012markov} is actually computationally intractable to run, hence we left it out of the benchmark for larger data sets. The observations made on SYNTHETIC750 (superiority of the \textsc{MAXIMIZATION} strategy over the \textsc{SAMPLING} one in terms of relevance) from those experiments are confirmed on larger synthetic data sets. Moreover, even if MMR might have better results on relevance-related metrics ($\text{rel}$, $\text{prec}$), this baseline clearly fails for diversity metrics ($\text{div}^\text{G}$, $\text{div}^{+}$), especially on larger data sets. As the number of items $\NITEMS$ increases, the difference in performance between algorithms from the DQD family decreases. This might be due to the fact that the more items there are, the less important is the impact of user history on the selection of items. We also applied recommender systems to MovieLens and PREDICT. We only show in these tables the results with the \textsc{MAXIMIZATION} strategy, on account of the observations on the synthetic data sets. 

 \begin{table}[t]
  \caption{Benchmark on MovieLens (4 movie streamers/users, $\NBATCH=3$) and PREDICT (9 diseases/users, $\NBATCH=3$).}
  \label{tab:compl_experiments2}
  \begin{center}
    \begin{small}
      \begin{sc}
        \begin{tabular}{lrrrrrr}
          \toprule
          MovieLens  & &  &   \\
        \textsc{SAMPLING} & \text{rel} $\uparrow$ & \text{prec} $\uparrow$ & \text{div}$^L$  $\uparrow$ & \text{div}$^G$  $\uparrow$ & \text{div}$^+$ $\uparrow$ & Time $\downarrow$  \\
          \midrule
QDDecomp.   & 1.01$\pm$ 0.02 &   0.01$\pm$ 0.01  &0.84$\pm$ 0.01 & \underline{0.05$\pm$ 0.00} &   0.02$\pm$ 0.02 & 70.96 $\pm$1.35\\
CondDPP  &  1.01$\pm$ 0.0   & 0.0$\pm$ 0.0 &  \underline{0.85$\pm$ 0.0} & \textbf{0.06$\pm$ 0.01}  & 0.01$\pm$ 0.00  & 2.70 $\pm$2.06\\
$\varepsilon$-Greedy      &    1.01$\pm$ 0.02  &  0.01$\pm$ 0.01  &0.84$\pm$ 0.01 & \underline{0.05$\pm$ 0.00} &   0.02$\pm$ 0.02 & 66.55 $\pm$1.55     \\     
\textbf{\ALGO}     &       1.01$\pm$ 0.01  &  0.0$\pm$ 0.01 &  \textbf{0.86$\pm$ 0.00} & \textbf{0.06$\pm$ 0.01} &   0.01$\pm$ 0.01 & 66.95 $\pm$1.67      \\    
        \midrule
        \textsc{MAXIMIZ.} & &  &   \\
QDDecomp.   & \underline{3.49$\pm$ 0.09} & \underline{0.76$\pm$ 0.04} & \underline{0.80$\pm$ 0.01} & \textbf{0.06$\pm$ 0.01} & 0.60$\pm$ 0.03& 25.13 $\pm$0.30  \\ 
CondDPP   &   3.01$\pm$ 0.12 &  0.65$\pm$ 0.04 & \textbf{0.83$\pm$ 0.01} & \textbf{0.06$\pm$ 0.01} & \underline{0.65$\pm$ 0.04} &25.72 $\pm$0.34             \\    
$\varepsilon$-Greedy       &    \underline{3.49$\pm$ 0.09} & \underline{0.76$\pm$ 0.04}  &\underline{0.80$\pm$ 0.01}  &\textbf{0.06$\pm$ 0.01} & 0.60$\pm$ 0.03  &  25.72 $\pm$0.34\\
\textbf{\ALGO}         & \underline{3.49$\pm$ 0.09}  &\underline{ 0.76 $\pm$0.04} & \underline{0.80 $\pm$0.01 } & \textbf{0.06$\pm$ 0.01} &  0.60$\pm$ 0.03 & 26.83 $\pm$0.34\\
        \midrule
MMR         &     \textbf{  3.78$\pm$ 0.08} & \textbf{0.84 $\pm$0.03} &  0.74 $\pm$0.02 &  \textbf{0.06$\pm$ 0.01} & \textbf{0.77$\pm$ 0.03}& 27.25 $\pm$0.35\\
 \midrule
  \midrule
          PREDICT & &  &   \\
        \textsc{SAMPLING} & \text{rel} $\uparrow$ & \text{prec} $\uparrow$ & \text{div}$^L$  $\uparrow$ & \text{div}$^G$  $\uparrow$ & \text{div}$^+$ $\uparrow$ & Time $\downarrow$  \\
          \midrule
QDDecomp.&0.51$\pm$ 0.01&0.52$\pm$ 0.02&0.68$\pm$ 0.01&0.44$\pm$ 0.02&0.77$\pm$ 0.03&1.02$\pm$ 0.01 \\ 
CondDPP&0.53$\pm$ 0.01&0.57$\pm$ 0.01&0.69$\pm$ 0.01&0.45$\pm$ 0.02&0.86$\pm$ 0.02&1.25$\pm$ 0.01 \\ 
$\varepsilon$-Greedy&0.51$\pm$ 0.01&0.52$\pm$ 0.02&0.68$\pm$ 0.01&0.44$\pm$ 0.02&0.77$\pm$ 0.03&1.01$\pm$ 0.0 \\ 
\textbf{\ALGO}&0.52$\pm$ 0.01&0.54$\pm$ 0.02&0.77$\pm$ 0.01&0.45$\pm$ 0.02&0.87$\pm$ 0.03&1.29$\pm$ 0.01 \\ 
        \midrule
        \textsc{MAXIMIZ.} & &  &   \\
QDDecomp.&\textbf{0.76$\pm$ 0.02}&\textbf{0.93$\pm$ 0.01}&\textbf{1.02$\pm$ 0.02}&0.57$\pm$ 0.02&\textbf{1.02$\pm$ 0.02}&0.87$\pm$ 0.01 \\ 
CondDPP&\underline{0.75$\pm$ 0.02}&\textbf{0.93$\pm$ 0.01}&\underline{0.87$\pm$ 0.01}&\textbf{0.63$\pm$ 0.02}&\underline{0.88$\pm$ 0.01}&1.07$\pm$ 0.01 \\ 
$\varepsilon$-Greedy&\textbf{0.76$\pm$ 0.02}&\textbf{0.93$\pm$ 0.01}&\textbf{1.02$\pm$ 0.02}&0.57$\pm$ 0.02&\textbf{1.02$\pm$ 0.02}&0.87$\pm$ 0.01 \\ 
\textbf{\ALGO}&\textbf{0.76$\pm$ 0.02}&\textbf{0.93$\pm$ 0.01}&\textbf{1.02$\pm$ 0.02}&0.57$\pm$ 0.02&\textbf{1.02$\pm$ 0.02}&1.15$\pm$ 0.01 \\
        \midrule
MMR&0.66 $\pm$0.02&\underline{0.86 $\pm$0.01}&0.69 $\pm$0.02&\underline{0.60 $\pm$0.02}&0.70$\pm$ 0.02&1.88$\pm$ 0.01 \\ 
          \bottomrule
        \end{tabular}
      \end{sc}
    \end{small}
  \end{center}
  \vskip -0.1in
\end{table}

 \begin{table}[t]
  \caption{Benchmark on SYNTHETIC750 (6 users, $\NBATCH=3$), SYNTHETIC30k (4 users, $\NBATCH=3$), SYNTHETIC3M (4 users, $\NBATCH=3$) and SYNTHETIC15M (2 users, $\NBATCH=5$).}
  \label{tab:compl_experiments}
  \begin{center}
    \begin{small}
      \begin{sc}
        \begin{tabular}{lrrrrrr}
          \toprule
          SYNTHETIC750  & &  &   \\
        \textsc{SAMPLING} & \text{rel} $\uparrow$ & \text{prec} $\uparrow$ & \text{div}$^L$  $\uparrow$ & \text{div}$^G$  $\uparrow$ & \text{div}$^+$ $\uparrow$ & Time $\downarrow$  \\
          \midrule
QDDecomp.&0.51$\pm$ 0.0&0.61$\pm$ 0.01&\textbf{1.01$\pm$ 0.0}&\underline{0.95$\pm$ 0.0}&\textbf{1.01$\pm$ 0.0}&0.06$\pm$ 0.0 \\ 
CondDPP&0.51$\pm$ 0.0&0.64$\pm$ 0.01&\underline{1.0$\pm$ 0.0}&0.94$\pm$ 0.01&0.97$\pm$ 0.01&0.07$\pm$ 0.0 \\ 
$\varepsilon$-Greedy&0.51$\pm$ 0.0&0.61$\pm$ 0.01&\textbf{1.01$\pm$ 0.0}&\underline{0.95$\pm$ 0.0}&1.01$\pm$ 0.0&0.06$\pm$ 0.0 \\ 
MarkovDPP&0.51$\pm$ 0.0&0.59$\pm$ 0.01&\underline{1.0$\pm$ 0.0}&0.17$\pm$ 0.0&0.91$\pm$ 0.01&0.08$\pm$ 0.0 \\ 
\textbf{\ALGO}&0.51$\pm$ 0.0&0.63$\pm$ 0.01&\textbf{1.01$\pm$ 0.0}&\underline{0.95$\pm$ 0.0}&\underline{0.98$\pm$ 0.01}&0.14$\pm$ 0.0\\
        \midrule
        \textsc{MAXIMIZ.} & &  &   \\
QDDecomp.&\textbf{0.61$\pm$ 0.0}&\textbf{1.0$\pm$ 0.0}&\textbf{1.01$\pm$ 0.0}&0.9$\pm$ 0.01&\textbf{1.01$\pm$ 0.0}&0.03$\pm$ 0.0\\ 
CondDPP&\underline{0.6$\pm$ 0.0}&\textbf{1.0$\pm$ 0.0}&\textbf{1.01$\pm$ 0.0}&0.89$\pm$ 0.01&\textbf{1.01$\pm$ 0.0}&0.04$\pm$ 0.0\\ 
$\varepsilon$-Greedy&\textbf{0.61$\pm$ 0.0}&\textbf{1.0$\pm$ 0.0}&\textbf{1.01$\pm$ 0.0}&0.9$\pm$ 0.01&\textbf{1.01$\pm$ 0.0}&0.03$\pm$ 0.0\\ 
MarkovDPP&0.58$\pm$ 0.0&\underline{0.94$\pm$ 0.01}&\textbf{1.01$\pm$ 0.0}&0.17$\pm$ 0.0&\underline{0.98$\pm$ 0.01}&0.06$\pm$ 0.0\\ 
\textbf{\ALGO}&\textbf{0.61$\pm$ 0.0}&\textbf{1.0$\pm$ 0.0}&\textbf{1.01$\pm$ 0.0}&\underline{0.95$\pm$ 0.0}&\textbf{1.01$\pm$ 0.0}&0.12$\pm$ 0.0\\ 
        \midrule
MMR&\underline{0.6$\pm$ 0.0}&\textbf{1.0$\pm$ 0.0}&\textbf{1.01$\pm$ 0.0}&\textbf{0.96$\pm$ 0.0}&\textbf{1.01$\pm$ 0.0}&0.08$\pm$ 0.0\\ 
 \midrule
  \midrule
          SYNTHETIC30k  & &  &   \\
        \textsc{SAMPLING} & \text{rel} $\uparrow$ & \text{prec} $\uparrow$ & \text{div}$^L$  $\uparrow$ & \text{div}$^G$  $\uparrow$ & \text{div}$^+$ $\uparrow$ & Time $\downarrow$  \\
          \midrule
QDDecomp.&0.5$\pm$ 0.0&0.33$\pm$ 0.0&\textbf{1.01$\pm$ 0.0}&0.96$\pm$ 0.0&\underline{1.0$\pm$ 0.0}&47.05$\pm$ 0.15 \\ 
CondDPP&0.5$\pm$ 0.0&0.33$\pm$ 0.0&\textbf{1.01$\pm$ 0.0}&0.96$\pm$ 0.0&\underline{1.0$\pm$ 0.0}&47.19$\pm$ 0.12 \\ 
$\varepsilon$-Greedy&0.5$\pm$ 0.0&0.33$\pm$ 0.0&\textbf{1.01$\pm$ 0.0}&0.96$\pm$ 0.0&\underline{1.0$\pm$ 0.0}&47.49$\pm$ 0.32 \\ 
\textbf{\ALGO}&0.5$\pm$ 0.0&0.33$\pm$ 0.0&\textbf{1.01$\pm$ 0.0}&0.96$\pm$ 0.0&\underline{1.0$\pm$ 0.0}&48.1$\pm$ 0.12 \\ 
        \midrule
        \textsc{MAXIMIZ.} & &  &   \\
QDDecomp.&\underline{0.54$\pm$ 0.0}&\underline{0.92$\pm$ 0.02}&\textbf{1.01$\pm$ 0.0}&0.96$\pm$ 0.0&\textbf{1.01$\pm$ 0.0}&0.79$\pm$ 0.0 \\ 
CondDPP&\underline{0.54$\pm$ 0.0}&\underline{0.92$\pm$ 0.02}&\textbf{1.01$\pm$ 0.0}&0.96$\pm$ 0.0&\textbf{1.01$\pm$ 0.0}&1.05$\pm$ 0.01 \\ 
$\varepsilon$-Greedy&\underline{0.54$\pm$ 0.0}&\underline{0.92$\pm$ 0.02}&\textbf{1.01$\pm$ 0.0}&0.96$\pm$ 0.0&\textbf{1.01$\pm$ 0.0}&0.79$\pm$ 0.0 \\  
\textbf{\ALGO}&\underline{0.54$\pm$ 0.0}&\underline{0.92$\pm$ 0.02}&\textbf{1.01$\pm$ 0.0}&0.96$\pm$ 0.0&\textbf{1.01$\pm$ 0.0}&1.45$\pm$ 0.0 \\ 
        \midrule
MMR&\textbf{0.68$\pm$ 0.0}&\textbf{1.0$\pm$ 0.0}&\underline{1.0$\pm$ 0.0}&0.96$\pm$ 0.0&\underline{1.0$\pm$ 0.0}&1.73$\pm$ 0.0 \\
          \midrule
          \midrule
          SYNTHETIC3M  & &  &   \\
        \textsc{SAMPLING} & \text{rel} $\uparrow$ & \text{prec} $\uparrow$ & \text{div}$^L$  $\uparrow$ & \text{div}$^G$  $\uparrow$ & \text{div}$^+$ $\uparrow$ & Time $\downarrow$  \\
          \midrule
QDDecomp.  & 0.5$\pm$ 0.0 &\underline{ 0.58$\pm$ 0.03} & \textbf{1.01$\pm$ 0.0} & \textbf{0.97$\pm$ 0.0 }& \textbf{1.01$\pm$ 0.0} &  248.8$\pm$ 10.44\\
CondDPP  &  0.5$\pm$ 0.0 & \underline{ 0.58$\pm$ 0.03}  & \textbf{1.01$\pm$ 0.0} & \textbf{0.97$\pm$ 0.0 } & \textbf{1.01$\pm$ 0.0}   &331.23$\pm$ 6.84\\
$\varepsilon$-Greedy    &  0.5$\pm$ 0.0 & \underline{ 0.58$\pm$ 0.03}  & \textbf{1.01$\pm$ 0.0} & \textbf{0.97$\pm$ 0.0 } & \textbf{1.01$\pm$ 0.0} &  235.06$\pm$ 4.81\\
\textbf{\ALGO}      &   0.5$\pm$ 0.0 & \underline{ 0.58$\pm$ 0.03}  & \textbf{1.01$\pm$ 0.0} & \textbf{0.97$\pm$ 0.0 } & \textbf{1.01$\pm$ 0.0} & 383.04$\pm$ 16.32\\
        \midrule
        \textsc{MAXIMIZ.} & &  &   \\
QDDecomp.   &   \underline{0.55$\pm$ 0.0}  &  \textbf{1.0$\pm$ 0.0} & \textbf{1.01$\pm$ 0.0} & \textbf{0.97$\pm$ 0.0 } & \textbf{1.01$\pm$ 0.0} &    232.0$\pm$ 4.95\\
CondDPP    & \underline{ 0.55$\pm$ 0.0 }  & \textbf{1.0$\pm$ 0.0 }& \textbf{1.01$\pm$ 0.0} & \textbf{0.97$\pm$ 0.0 } & \textbf{1.01$\pm$ 0.0} & 297.51$\pm$ 6.56\\
$\varepsilon$-Greedy     & \underline{  0.55$\pm$ 0.0}  & \textbf{1.0$\pm$ 0.0 } & \textbf{1.01$\pm$ 0.0}&  \textbf{0.97$\pm$ 0.0 } & \textbf{1.01$\pm$ 0.0}&   224.08$\pm$ 7.87\\
\textbf{\ALGO}      & \underline{ 0.55$\pm$ 0.0 }  & \textbf{1.0$\pm$ 0.0 } & \textbf{1.01$\pm$ 0.0} & \textbf{0.97$\pm$ 0.0 } & \textbf{1.01$\pm$ 0.0} & 372.8$\pm$ 18.67\\
        \midrule
MMR         &    \textbf{0.73$\pm$ 0.0} &  \textbf{ 1.0$\pm$ 0.0} & \underline{0.02$\pm$ 0.0 } & \underline{0.0$\pm$ 0.0} &  \underline{0.02$\pm$ 0.0}  &569.72$\pm$ 12.82\\
\midrule
\midrule
SYNTHETIC15M & & & & \\
        \textsc{MAXIMIZ.} & \text{rel} $\uparrow$ & \text{prec} $\uparrow$ & \text{div}$^L$  $\uparrow$ & \text{div}$^G$  $\uparrow$ & \text{div}$^+$ $\uparrow$ & Time $\downarrow$  \\
          \midrule
CondDPP & \underline{0.53$\pm$ 0.0} & \underline{0.8$\pm$ 0.04} & \textbf{1.0$\pm$ 0.0} & \textbf{0.94$\pm$ 0.0} & \textbf{1.01$\pm$ 0.0}   & 63.03$\pm$ 0.58\\
\textbf{\ALGO}     &    \underline{0.53$\pm$ 0.0} & \underline{0.8$\pm$ 0.04} & \textbf{1.0$\pm$ 0.0} & \textbf{0.94$\pm$ 0.0} & \textbf{1.01$\pm$ 0.0} & 251.66$\pm$ 0.73\\
        \midrule
MMR      &       \textbf{0.73$\pm$ 0.0 } & \textbf{1.0$\pm$ 0.0}  &  \underline{0.0$\pm$ 0.0} &   \underline{0.0$\pm$ 0.0} &   \underline{0.0$\pm$ 0.0} & 98.2$\pm$ 0.21\\
          \bottomrule
        \end{tabular}
      \end{sc}
    \end{small}
  \end{center}
  \vskip -0.1in
\end{table}

\clearpage

\section{About MovieLens and diversity}\label{app:discussion_movielens}

To propose a more quantitative analysis of the performance of \ALGO{} and MMR on the MovieLens data set, we compute several metrics on all tested real-life data sets in Table~\ref{tab:datasets_metrics}. For each metric, we also report the ratio of the metric value for our contribution \ALGO{} over the metric value for MMR. For instance, using the values from Table~\ref{tab:compl_experiments2}, we obtain a REL ratio of $\frac{3.48}{3.78}=0.92$, where $3.78$ is the relevance value achieved by MMR and $3.48$ is the relevance value achieved by B-DivRec. We denote in bold type the cases where the ratio is greater or equal to 1, meaning similar performance or superiority of \ALGO{} over MMR.

As a proxy for the implicit feedback bias, we compute the sparsity number, meaning the percentage of observed feedback (that is, usually the number of non-zero values in the user-item rating matrix). The smaller the sparsity number, the greater the bias, as the feedback model is then trained on a smaller number of observed values. 

We measure the popularity bias by the Gini coefficient~\citep{braun2023metrics}, denoted Gini in the table, that quantifies inequality in the distributions of item popularity scores. The score is in the range [0,1], where 0 represents no bias, and the larger, the more popularity bias is present.

We also compute metrics regarding the diversity of the user history and data set. History-wise diversity (Hist-Div) represents the global diversity across items in the user history, averaged across users, that is, the collinearity of embeddings of items in the history. History-wise diversity for a given user is DIV$^G$ computed when $S^t=\emptyset$ in Equation~\ref{eq:def_metrics}. Intrinsic diversity (Diversity) is the volume across all items in the data set. We choose to set the volume of an empty set to $0$, when the user history is empty.

We also tested new, richer embeddings of movies in MovieLens using the MiniLM 384 model~\citep{wang2020minilm} instead of the Universal Sentence Encoder model~\citep{cer2018universal} to generate embeddings from the movie title and keywords. We name the corresponding data set RicherMovieLens, and we run it with the same parameters as MovieLens in the main text.

Note that Fdataset and Gottlieb have the same rating matrix, but features in Gottlieb are richer (as evidenced by the history-wise diversity value). As mentioned in the paper, in MovieLens, the history-wise diversity is close to 0, meaning that there are at least two embeddings in the history of almost each user which are collinear. This impairs the computation of the score (Equation~\ref{eq:qdt_metric}) for \ALGO, whereas this effect is lessened in MMR, as MMR only relies on the maximum of pairwise diversity scores (Section~\ref{sec:related_works}). However, as illustrated by the performance ratios, taking into account diversity over sets incurs a higher diversity in recommendations. 

Our assumption is that the poor history-wise diversity metrics on MovieLens are mostly due to collinear item embeddings, where extremely similar movies (like successive entries in a series) have almost the same embedding, collapsing the volume-based diversity metric. A possible solution to mitigate this issue is to consider a "representative" set of items in the user history, instead of the whole history which might contain collinear item embeddings. In Section~\ref{app:related_works}, we discuss at length (ridge) leverage scores, which can be used to determine representative points in a set. Representative points are points which are the most "unique", that is, decorrelated from other points.

\begin{table}[t]
  \caption{RicherMovieLens and MovieLens (4 users, $\NBATCH=3$, $\tau=2.5$) with \textsc{MAXIMIZATION}. The results for MovieLens are the same as in Table~\ref{tab:benchmark}.}
  \label{tab:richer_movielens}
  \begin{center}
    \begin{small}
      \begin{sc}
        \begin{tabular}{lllllll}
          \toprule
         MovieLens  & \text{rel} $\uparrow$        & \text{div}$^L$ $\uparrow$    & \text{div}$^G$ $\uparrow$  & \text{div}$^+$ $\uparrow$  & Time $\downarrow$ \\
          \midrule
MMR&     \textbf{  3.78$\pm$ 0.08} & 0.74 $\pm$0.02 &  \textbf{0.06$\pm$ 0.01} & \textbf{0.77$\pm$ 0.03}& 27.25 $\pm$0.35\\
\textbf{\ALGO}& 3.49$\pm$ 0.09  & \textbf{0.80 $\pm$0.01 } & \textbf{0.06$\pm$ 0.01} &  0.60$\pm$ 0.03 & \textbf{26.83 $\pm$0.34}\\
\midrule
\midrule
         RicherMovieLens   &  &  &  &  & \\
         \midrule
MMR  &  \textbf{3.91$\pm$ 0.07} & 0.72$\pm$ 0.02 & 0.05$\pm$ 0.0 & \textbf{0.76$\pm$ 0.02} & \textbf{19.00$\pm$ 0.01}\\
\textbf{\ALGO}  & 3.68$\pm$ 0.06 &\textbf{0.81$\pm$ 0.0} & 0.05$\pm$ 0.0 & 0.73$\pm$ 0.03 & 19.31$\pm$ 0.03\\
          \bottomrule
        \end{tabular}
      \end{sc}
    \end{small}
  \end{center}
  \vskip -0.1in
\end{table}

 \begin{table}[t]
  \caption{Metrics related to known bias in recommender systems on real-life data sets, and relative performance of \ALGO{} compared to MMR.}
  \label{tab:datasets_metrics}
  \begin{center}
    \begin{small}
      \begin{sc}
        \begin{tabular}{lllllllll}
          \toprule
         Data set  & MovieLens        & PREDICT    & Epinions & Gottlieb & Fdataset & Cdataset & DNdataset\\
          \midrule
Sparsity (\%) & 1.50 & 0.50 & 0.17 & 1.18  & 1.18 &1.09& 0.03\\ 
Gini & 0.94 &  0.96 & 0.35 & 0.96 & 0.96 & 0.96& 1.00 \\ 
Hist-Div &  0.0 & 0.61 & 0.03 & 0.42 & 0.28& 0.29& 0.25\\
Diversity & 0.0 & 0.0 & 0.0 & 0.0 & 0.0& 0.0& 0.0\\
REL ratio & 0.92 & \textbf{1.15} & 0.61 & 0.87 & 0.86 & 0.87& 0.94\\
DIV$^L$ ratio & \textbf{1.08} & \textbf{1.48} & \textbf{1.00} & \textbf{1.19} & \textbf{1.82}& \textbf{1.39}& \textbf{1.13} \\
DIV$^G$ ratio & \textbf{1.00} & 0.95 & \textbf{1.00} & \textbf{1.09} & \textbf{1.52}& \textbf{1.54}& \textbf{1.13}\\
DIV$^+$ ratio & 0.78 & \textbf{1.46} & \textbf{1.00} & \textbf{1.21} & \textbf{1.82} & \textbf{1.42}& \textbf{1.13}\\
\midrule
\midrule
         Data set  & LRSSL       & RicherMovieLens   &  &  &  &  & \\
         \midrule
Sparsity    &  0.72 & 1.36 &  &  &  &  & \\
Gini  & 0.97 &0.94 &  &  &  &  & \\
Hist-Div   & 0.53 & 0.0 &  &  &  &  & \\
Diversity   & 0.0 & 0.0 &  &  &  &  & \\
REL ratio   & 0.86 & 0.94 &  &  &  &  & \\
DIV$^L$ ratio  & \textbf{1.05} & \textbf{1.13} &  &  &  &  & \\
DIV$^G$ ratio   & \textbf{1.01} & 0.98 &  &  &  &  & \\
DIV$^+$ ratio   & \textbf{1.05} &  0.97 &  &  &  &  & \\
          \bottomrule
        \end{tabular}
      \end{sc}
    \end{small}
  \end{center}
  \vskip -0.1in
\end{table}

\end{document}